\documentclass[journal]{IEEEtran}
\usepackage{amsmath,amsfonts}
\usepackage{algorithmic}
\usepackage{algorithm}
\usepackage{array}
\usepackage{bm}
\usepackage{textcomp}
\usepackage{stfloats}
\usepackage{url}
\usepackage{verbatim}
\usepackage{graphicx} 
\usepackage{cite}
\usepackage{amssymb}
\usepackage{CJK}
\usepackage{indentfirst}
\usepackage{cases}
\usepackage{setspace}
\usepackage{hyperref}
\usepackage{float}
\usepackage{titlecaps}
\usepackage{acronym}

\acrodef{isac}[ISAC]{integrated sensing and communications}
\acrodef{sinr}[SINR]{signal-to-interference-plus-noise ratio}
\acrodef{crb}[CRB]{Cram{\'e}r-Rao Bound}
\acrodef{dof}[DoF]{degrees of freedom}
\acrodef{miso}[MISO]{ multiple-input single-output}
\acrodef{mimo}[MIMO]{multiple-input and multiple-output}
\acrodef{mu-mimo}[MU-MIMO]{multi-user multiple-input and multiple-output}
\acrodef{mi}[MI]{mutual information}
\acrodef{pwm}[PWM]{planar-wave model}
\acrodef{swm}[SWM]{spherical-wave model}
\acrodef{nusw}[NUSW]{non-uniform spherical-wave}
\acrodef{hspm}[HSPM]{hybrid spherical and planar wavefront model}
\acrodef{aoa}[AoA]{angle of arrival}
\acrodef{aod}[AoD]{angle of departure}
\acrodef{hbf}[HBF]{hybrid beamforming}
\acrodef{los}[LoS]{line-of-sight}
\acrodef{nlos}[NLoS]{non-line-of-sight}
\acrodef{awgn}[AWGN]{additive white Gaussian noise}
\acrodef{rcs}[RCS]{radar cross section}
\acrodef{fim}[FIM]{Fisher’s information matrix}
\acrodef{scnr}[SCNR]{signal-to-clutter-plus-noise ratio}
\acrodef{mvdr}[MVDR]{minimum variance distortionless response}
\acrodef{svd}[SVD]{singular-value-decomposition}
\acrodef{qmp}[QMP]{quadratic matrix programming}
\acrodef{sdr}[SDR]{semidefinite relaxation}
\acrodef{sdp}[SDP]{semidefinite programming}
\acrodef{kkt}[KKT]{Karush-Kuhn-Tucker}
\acrodef{music}[MUSIC]{mutiple signal 
classification}
\acrodef{xl}[XL]{extremely
large-scale}
\acrodef{mmwave}[mmWave]{millimeter wave}
\acrodef{thz}[THz]{terahertz}
\acrodef{xl-mimo}[XL-MIMO]{extremely
large-scale multiple-input and multiple-output}
\acrodef{wsms}[WSMS]{widely-spaced multi-subarray}
\acrodef{rf}[RF]{radio frequency}
\acrodef{bs}[BS]{base station}
\acrodef{hda}[HDA]{hybrid digital and analog}
\acrodef{tx}[Tx]{transmitter}
\acrodef{rx}[Rx]{receiver}
\acrodef{rf}[RF]{radio frequency}
\acrodef{iid}[i.i.d.]{independent and identically distributed}

\ifCLASSOPTIONcompsoc
\usepackage[caption=false,font=normalsize,labelfont=sf,textfont=sf]{subfig}
\else
\usepackage[caption=false,font=footnotesize]{subfig}
\fi

\hypersetup{colorlinks=true}
\def\BibTeX{{\rm B\kern-.05em{\sc i\kern-.025em b}\kern-.08em
		T\kern-.1667em\lower.7ex\hbox{E}\kern-.125emX}}

\newtheorem{remark}{\textbf{Remark}}
\newtheorem{theorem}{\textbf{Theorem}}

\newtheorem{lemma}{\textbf{Lemma}}

\newtheorem{corollary}{\textbf{Corollary}}

\newtheorem{proposition}{\textbf{Proposition}}

\newenvironment{proof}{{\indent \indent \it Proof:\quad}}{\hfill $\blacksquare$\par}

\makeatletter

\newcommand{\Rmnum}[1]{\expandafter\@slowromancap\romannumeral #1@}
\makeatother

\usepackage{framed} 

\usepackage{cancel}

\begin{document}

	\title{Near-Field Hybrid Beamforming Design for  Modular XL-MIMO ISAC Systems}
	\author{Chunwei~Meng, Dingyou~Ma, Zhaolin~Wang, Yuanwei~Liu, Zhiqing~Wei, Zhiyong~Feng
		\thanks{C. Meng, D. Ma, Z. Wei, and Z. Feng are with the Key Laboratory of Universal Wireless Communications, Ministry of Education, Beijing University of Posts and Telecommunications, Beijing 100876, China (e-mail: \{mengchunwei, dingyouma, weizhiqing, fengzy\}@bupt.edu.cn).
  Z. Wang and Y. Liu are with the School of Electronic Engineering and Computer Science, Queen Mary University of London, London E1 4NS, U.K. (e-mail: \{zhaolin.wang, yuanwei.liu\}@qmul.ac.uk).
  (Corresponding author: Zhiqing~Wei,  Zhiyong~Feng)
		}
		
	}
	\maketitle
	\begin{abstract}
		\textcolor{black}{A novel modular extremely large-scale multiple-input-multiple-output (XL-MIMO) integrated sensing and communication (ISAC) system is investigated in this paper.
        A downlink sensing scenario is considered, where the modular XL-MIMO BS, consisting of widely-spaced identical subarrays, simultaneously serves a multi-antenna user and performs target sensing in the presence of multiple interferences.
        Due to the small subarray aperture and large array spacing, a piecewise-far-field channel model is employed to characterize both communication and sensing channels, where far-field propagation exists within each subarray while near-field effects dominate among subarrays.
       Based on the group-connected hybrid architecture, a joint analog-digital beamforming problem is formulated to optimize communication spectral efficiency while satisfying the sensing signal-to-clutter-plus-noise ratio (SCNR) requirement.
       Then, the optimal transmit covariance matrix is proved to lie in the subspace spanned by subarray response vectors.
       Based on this, a closed-form solution for the optimal analog beamformer is derived, and the original joint analog-digital optimization problem is transformed into a low-dimensional rank-constrained digital beamforming problem.
       Firstly, the semi-closed form of the optimal digital beamformer is derived and shown to form a complex Stiefel manifold.
       Based on this structure, a joint Riemannian-Euclidean gradient descent algorithm is developed for iterative optimization.
      Then,  an semidefinite relaxation-based approach is proposed, where the near-optimal solution is  obtained through rank constraint relaxation and randomization.
      The superiority of the proposed
algorithms is extensively validated through, revealing that the optimal subarray scale balances spatial multiplexing and beamforming gains based on user distance, while increasing subarray numbers significantly enhances range resolution due to more pronounced spherical wavefronts.}
        %

	\end{abstract}
	

	\section{Introduction}

\IEEEPARstart{T}{he} sixth generation wireless systems (6G)
have been envisioned as a vital enabler for
numerous emerging applications, such as intelligent manufacturing and smart transportation\cite{liu2022isactoward}.
The challenging problem is to satisfy the requirements of these applications for high-capacity communications and high-resolution sensing, which motivates the development of \ac{isac} technologies\cite{lu2024isacsurvey}.
The advant of \ac{xl-mimo} and the exploration of \ac{mmwave}/sub-\ac{thz} frequency bands in 6G systems lead to a growing convergence between  communication and sensing in terms of channel characteristics and signal processing techniques.
This convergence makes it feasible to achieve high-precision sensing and high-speed communication simultaneously using integrated waveforms and hardware platforms, offering several advantages such as reduced hardware costs, improved spectral efficiency, and mutual benefits between sensing and communication functionalities\cite{hua2024nefisac,liu2019hybridisac}.

However, the large array aperture and high operating frequencies adopted in 6G systems significantly extend the Rayleigh distance, making the near-field effect more prominent for communication users and sensing targets\cite{liu2023tutorial}.
Recent studies have shown that near-field propagation not only invalidates traditional far-field channel models  but also introduces new opportunities and challenges for \ac{isac} techniques\cite{liu2024sensing, hua2024nefisac,cong2023near,wang2023nearisac,Galappaththige2024nfisac}.
%
The work in \cite{cong2023near} revealed that the near-field effect can potentially enhance both communication and sensing performance.
The authors in \cite{wang2023nearisac} proposed a novel near-field \ac{isac} framework that leverages the additional distance dimension for optimal waveform design, demonstrating superior performance compared to traditional far-field \ac{isac} systems. 
Subsequently,  \cite{Galappaththige2024nfisac} proposed an efficient iterative near-field beamforming algorithm for multi-target detection,  achieving significant enhancements in localization accuracy over conventional far-field techniques. 
Despite these advancements, the high computational complexity and substantial hardware deployment costs associated with near-field \ac{xl-mimo} \ac{isac} systems pose significant challenges for practical implementation in realistic scenarios \cite{cong2023near}.

To address these challenges, a novel modular \ac{xl-mimo} architecture, also known as \ac{wsms}, has recently emerged as a promising solution\cite{yan2021joint,li2022nearfield,li2023multi, lu2023tutorial,yang2023performance}. 
This architecture consists of multiple modular subarrays, each composed of a flexible number of array elements with typical half-wavelength spacing, while the subarrays are separated by relatively large distances\cite{li2022nearfield}.
The modular \ac{xl-mimo} architecture offers advantages such as reduced hardware costs, lower power consumption, and ease of flexible deployment in practical scenarios \cite{yang2023performance}.
Moreover, the enlarged array aperture resulting from the wide spacing between subarrays  exhibits a larger
near-field region, potentially enhancing the performance of near-field \ac{isac} systems.
Studies have shown that modular \ac{xl-mimo}, compared to collocated arrays, exhibits a more pronounced near-field effect \cite{li2022nearfield},  provides both inter-path and intra-path multiplexing gains to improve spectral efficiency\cite{yan2021joint}, and better adapts to the spatial non-stationarity of the channel\cite{chen2023non}.
In \cite{yang2023performance}, the authors investigated the potential of near-field localization with modular \ac{xl-mimo} by analyzing the \ac{crb} for angle and range estimation, revealing that the increased array aperture and angular span of the modular array significantly enhance the near-field localization performance compared to traditional collocated arrays.

Despite the benefits of modular \ac{xl-mimo} for both near-field communication and sensing, the design and optimization of modular \ac{xl-mimo} \ac{isac} systems remain unexplored.
The primary challenge lies in the hybrid beamforming design, where the group-connected hybrid architecture imposes a block-diagonal constraint on the analog beamformer, resulting in an intractable non-convex optimization problem\cite{zhangjun2020hybrid}.
Existing \ac{hbf} methods \cite{shen2022alter,wang2022partial} mainly focus on optimizing for a single communication functionality and cannot effectively handle the block-diagonal constraint in modular \ac{xl-mimo} \ac{isac} systems, where the coupled communication and sensing performance lead to potential conflicts and make it impractical to separately optimize each diagonal submatrix or directly apply iterative optimization methods.
%
Moreover, the extremely large antenna arrays introduce high computational complexity, making conventional optimization algorithms impractical for real-world implementation. 
Additionally, the modular structure exhibits unique channel propagation characteristics, with the  entire large-scale array operating in the near-field while the small-scale subarrays are in the far-field\cite{cui2024nearwide}.
This hybrid propagation characteristic  presents an opportunity for complexity reduction, as it naturally decomposes into far-field propagation within each subarray and near-field effects between subarrays, mitigating the nonlinear phase variation of received signals\cite{shin2024iccsubarray}. 
Therefore, it is necessary to develop an efficient optimization algorithm that exploits the hybrid propagation channel structure while addressing both the structural constraints and performance requirements.

In this paper,  we propose a low-complexity hybrid beamforming design tailored for modular \ac{xl-mimo} \ac{isac} systems.
 The main contributions of the paper are summarized as follows:
 \begin{itemize}
      \item We consider a downlink \ac{isac} scenario, where a \ac{bs}, equipped with modular \ac{xl}-arrays and the \ac{hda} architecture, serves a multi-antenna communication user while sensing a target in the presence of multiple interferences.
      The number of active \ac{rf} chains is adapted to the communication and sensing channel conditions to provide sufficient spatial \ac{dof}s.
      Considering the relatively small size of the subarrays compared to the entire array, we employ  the piecewise-far-field channel model to characterize the communication and sensing channels, as the user and target are typically  located in the far-field of each subarray but the near-field of the entire array.

	\item 
       \textcolor{black}{We formulate a joint analog-digital beamforming optimization problem based on communication spectral efficiency and sensing \ac{scnr}.
       Then, we derive a closed-form solution for the optimal analog beamformer by proving that the optimal transmit covariance matrix lies in the subspace spanned by the communication and sensing subarray response vectors.
       This theoretical insight transforms the original joint analog-digital optimization problem into a low-dimensional digital beamforming problem, substantially reducing the computational complexity.}
   \item 
      \textcolor{black}{We develop two efficient algorithms to solve the rank-constrained digital beamforming optimization. 
      First, we derive the semi-closed form of the optimal digital beamformer, which forms a complex Stiefel mainifold.
      Based on this, we propose a Riemannian-Euclidean joint gradient descent algorithm to iteratively obtain the local optimal solution.
      Additionally, we develop a \ac{sdr}-based approach that relaxes the non-convex rank constraint and transforms the problem into a \ac{sdp} one, which can be solved to obtain a near-optimal solution through the randomization technique.}

   \item 
 We conduct extensive simulations to validate the effectiveness of the proposed modular \ac{xl-mimo} \ac{isac}  algorithms.
   \textcolor{black}{The results reveal several key findings: i) The proposed \ac{sdr}-based algorithm achieves superior performance over conventional \ac{hbf} algorithms.
   ii) The optimal subarray scale varies with user distance, representing a tradeoff between spatial multiplexing and beamforming gains.
   iii) Increasing the number of subarrays significantly enhances the range resolution due to more pronounced spherical wavefronts across subarrays.}
   These insights provide valuable design guidelines for modular \ac{xl-mimo}  \ac{isac} systems.
   
\end{itemize}

The rest of this paper is organized as follows.
Section~\ref{section_system_model} introduces the system settings and the  transmmitted signal model. 
Section~\ref{section_performetric_problem} introduces communication and sensing channel models, along with their respective performance metrics, and presents the problem formulation. 
Section~\ref{algorithm_design} proposes two low-complexity \ac{hbf} algorithms.  
Simulation results are presented in Section~\ref{section_simulation}, followed by concluding remarks in Section~\ref{section_conclusion}.

\textit{Notations}:
We use boldface lower-case and upper-case letters to denote column vectors and  matrices, respectively.
$\mathbb{C}$ denotes the set of complex number  and $(\cdot)^T$, $(\cdot)^H$, $(\cdot)^{*}$ , $(\cdot)^{-1}$, $(\cdot)^{\dagger}$ denote the transposition, conjugate transposition, conjugate, inverse and pseudo-inverse,  respectively.
$\mathbb{E}(\cdot)$ denotes expectation.
$\otimes$ denotes the Kronecker product.
$\mathbf{A} \succeq \mathbf{0}$ indicates that the matrix $\mathbf{A}$ is positive semi-definite.
$\mathbf{I}_M$ indicates an $M\times M$ identity matrix.
The $\mathrm{diag}(\cdot)$ operator creates a diagonal matrix with the input elements placed along the main diagonal, while all other entries are set to zero.
$\det(\cdot)$ and $\text{tr}(\cdot)$ denote the determinant and trace of a matrix, respectively.

\section{System Model} \label{section_system_model}
	As illustrated in Fig. \ref{system_model}, we consider a downlink \ac{isac} system, where the \ac{isac} \ac{bs} equipped with  transmit and receive modular \ac{xl}-arrays  communicates with a multi-antenna user while simultaneously sensing a radar target.
Both  transmit and receive arrays are deployed along the $x$-axis, symmetrically centered around the origin, and widely separated in space to suppress signal leakage from transmitter and receive  clear sensing echoes\cite{yuan2021bayesian}.
	The total number of transmit and receive antennas is $N=KM$, with $K$ subarrays and $M$ antenna elements within each subarray.
	The communication user has $N_c$ antennas, where $N_c \ll N$.

    The antenna spacing within each subarray is $d=\frac{\lambda}{2}$, where $\lambda$ is the wavelength.
    For each subarray, we select its first element as the reference antenna. 
	We denote the  inter-subarray spacing $d_s=\Gamma d$ as the distance between the reference antennas of the adjacent subarrays, where $\Gamma \geq M$.
	There are two main reasons for considering $d_s$ to be much larger than $Md$.
	Firstly, it accommodates practical mounting structures, such as modular \ac{xl}-arrays mounted on building facades that are separated by windows\cite{li2023multi}.
	Secondly,  a large inter-subarray spacing expands the array aperture, $S$, which increases the near-field range of the overall antenna array to $\frac{2S^2}{\lambda}$.
   This expansion allows users initially located in the far-field of a collocated array to be within the near-field range of the modular array, enabling them to benefit from various advantages.

Assuming the distance between the transmit and receive arrays is $2D_0$, the position of the $m$-th element in the $k$-th subarray at \ac{tx} and \ac{rx} is given by $\mathbf{l}_{k, m}^t=({x}_{k, m}^t,0)$ and $\mathbf{l}_{k, m}^r=({x}_{k, m}^r,0)$, respectively, where ${x}_{k, m}^t={D_0}+(k-1)d_s+(m-1)d$ and ${x}_{k, m}^r=-{D_0}-(k-1)d_s-(m-1)d$, with $m \in \mathcal{M}\triangleq\{1,2, \ldots, M\}$ and $k  \in \mathcal{K}\triangleq\{1,2, \ldots, K\}$. 
	Suppose that a user, target, scatterer or interference located at $\mathbf{l}_{q}=(r \sin \theta, r \cos \theta)$is characterized by its distance $r$ from the origin and angle $\theta \in\left[-\frac{\pi}{2}, \frac{\pi}{2}\right]$ with respect to the positive $y$-axis.

	\begin{figure} [t]
 \setlength{\abovecaptionskip}{-0.1cm}
		\centering
		\includegraphics[width=0.4 \textwidth]{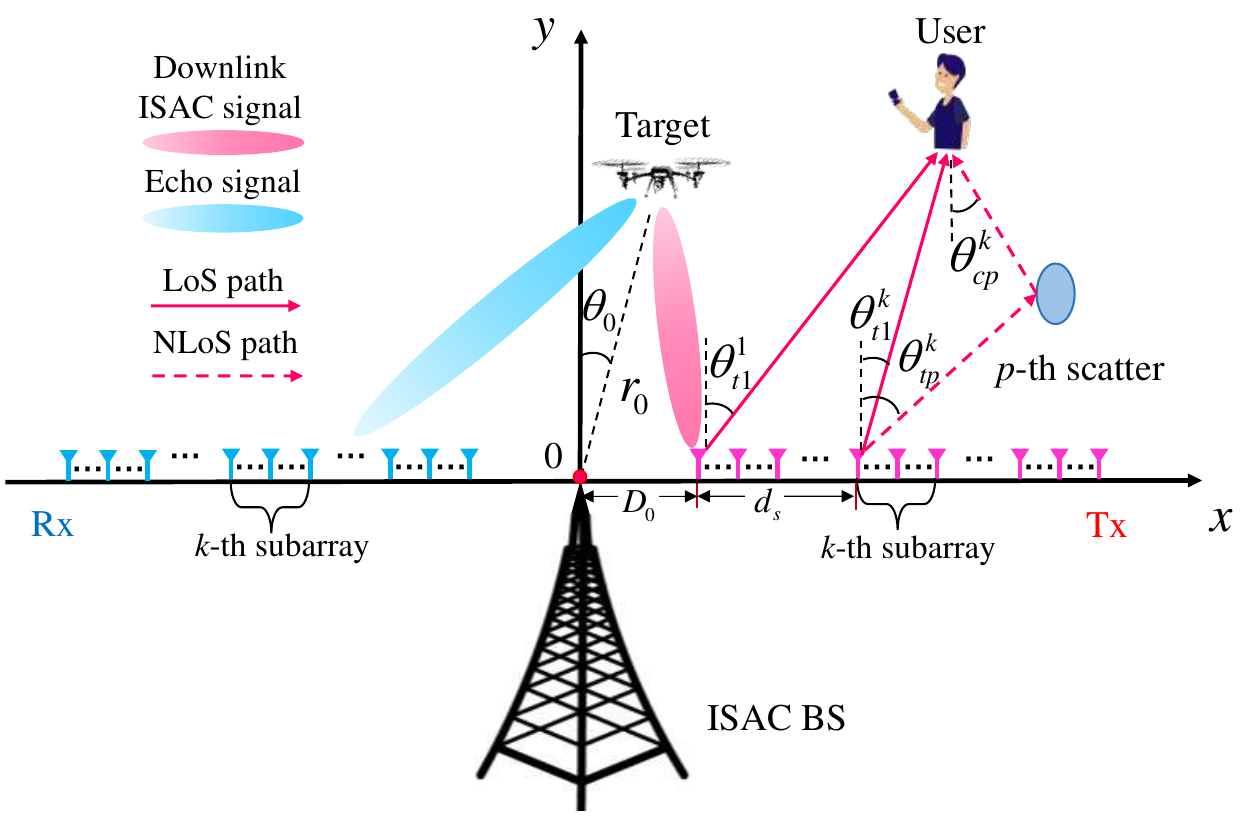}
		\caption{\textcolor{black}{Illustration of the downlink \ac{isac}  scenario  where the BS equipped with modular XL-arrays serves a multi-antenna communication user while sensing a target.}}
\label{system_model}\vspace{-3mm}
	\end{figure}

	\subsection{Transmit ISAC Signal}
	%

%

%

 As illustrated in Fig.~\ref{transmitter}, we consider a group-connected \ac{hda} architecture,  where each subarray is equipped with dedicated \ac{rf} chains that are connected to all antennas within that subarray but cannot be shared across subarrays due to the the hardware limitations\cite{yan2021joint,zhangjun2020hybrid}.
 Specifically, each subarray is equipped with $\tilde{M}_{\text{RF}}$ \ac{rf} chains, where each chain is connected to all antennas within that subarray, resulting in a total of $\tilde{N}_{\text{RF}}=K\tilde{M}_{\text{RF}}$ available \ac{rf} chains at the \ac{tx}.
 For energy efficiency, digital switches are employed after the digital precoder to adaptively control the number of active \ac{rf} chains according to the practical communication and sensing channel conditions \cite{zhang2024dynamic}.
We denote the number of active \ac{rf} chains as $M_{\text{RF}}$ per subarray, resulting in a total of $N_{\text{RF}}=KM_{\text{RF}}$ active RF chains.
For data transmission, the multi-antenna user can  can support the transmission of $N_s$ data streams, where $N_s \leq N_{\text{RF}}$.

We consider a coherent time block of $L$ symbols, during which the communication channels and sensing target parameters remain invariant.
 The narrowband transmitted \ac{isac} signal is denoted by $\mathbf{X}\triangleq[\mathbf{x}[1], \mathbf{x}[2], \ldots,\mathbf{x}[L]]\in\mathbb{C}^{N \times L}$,  where  $\mathbf{x}[l]$ is the transmitted signal at time index $l$.
 The information symbol matrix is given by $\mathbf{S} = [\mathbf{s}[1], \mathbf{s}[2], \ldots,\mathbf{s}[L]]\in {\mathbb{C}}^{N_s \times L}$, where  $\mathbf{s}[ l ]$ is the symbol vector transmitted at time index $l$, and $N_s$ is the number of data streams.
Therefore, the discrete-time transmitted signal at time index $l$ is expressed as
	\begin{equation} \label{transmit signal}
		\mathbf{x}[l]={\mathbf{W}}_{\text{RF}}{\mathbf{W}}_{\text{BB}}\mathbf{s}[l],
	\end{equation}
	where  $\mathbf{W}_{\text{RF}} \in \mathbb{C}^{N \times N_{\text{RF}}} $ and ${{\mathbf{W}}_{\text{BB}}}\text{=}\left[ {{\mathbf{w}}_{\text{BB},1}}, \ldots,{{\mathbf{w}}_{\text{BB},{{N}_{s}}}} \right]\in {{\mathbb{C}}^{{{N}_{\text{RF}}}\times {{N}_{s}}}}$ denote the analog and digital beamformers, respectively.
	Each entry in $\mathbf{S}$ is assumed to be i.i.d. and Gaussian distributed with zero mean and unit variance.
	The data streams are assumed to be independent of each other, satisfying $\mathbb{E}\left\{\mathbf{s}[l]\mathbf{s}^H[l]\right\}=\mathbf{I}, \forall l$\cite{wang2023nearisac}.
	Then, the covariance matrix of transmitted signal  is given by
	\begin{equation} \label{cov_matrix_RX}
		\mathbf{R}_{X} = \mathbb{E}\left\{\mathbf{x}[l]\mathbf{x}^H[l]\right\} ={{\mathbf{W}}_{\text{RF}}}{{\mathbf{W}}_{\text{BB}}}\mathbf{W}_{\text{BB}}^{H}\mathbf{W}_{\text{RF}}^{H}.
	\end{equation}

	\begin{figure} [t]
 \setlength{\abovecaptionskip}{-0.1cm}
		\centering
		\includegraphics[width=0.35 \textwidth]{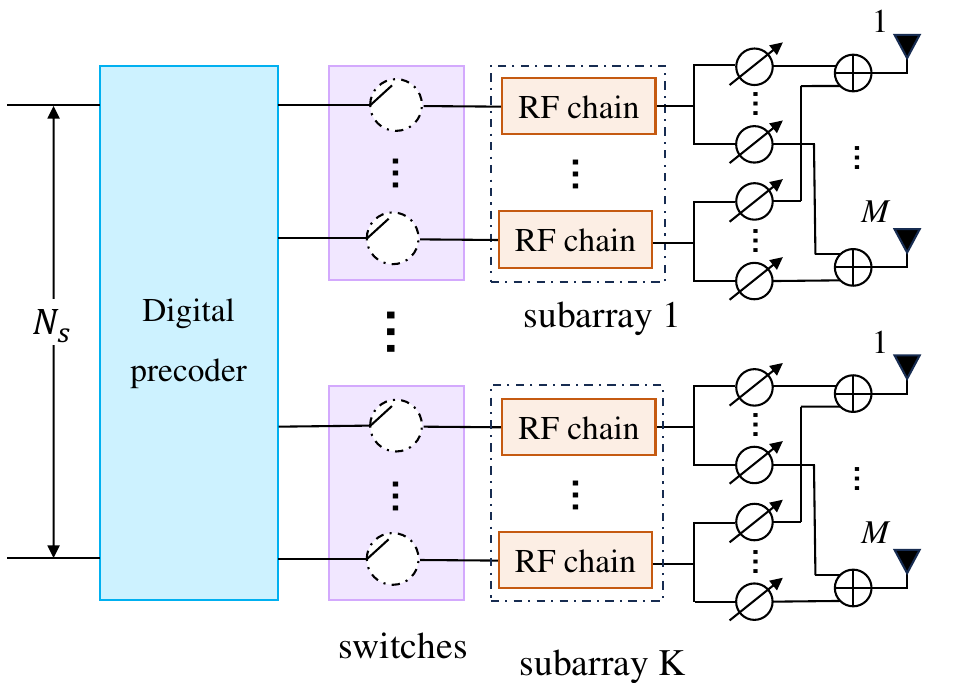}
		\caption{\textcolor{black}{Illustration of the hybrid beamforming architecture for the proposed modular XL-MIMO system.}}
		\label{transmitter}\vspace{-3mm}
	\end{figure}

The analog beamformer ${{\mathbf{W}}_{\text{RF}}}$ exhibits a block diagonal structure:
	\begin{equation}
           \setlength{\abovedisplayskip}{2pt}
	\setlength{\belowdisplayskip}{4pt}
{{\mathbf{W}}_{\text{RF}}}=\mathrm{diag}\left( \mathbf{\tilde{W}}_{\text{RF}}^{1},\cdots ,\mathbf{\tilde{W}}_{\text{RF}}^{K} \right),
	\end{equation}
	where \textcolor{black}{$\mathbf{\tilde{W}}_{\text{RF}}^{k}=\left[
		{{\mathbf{w}}_{k1}}, {{\mathbf{w}}_{k2}}, \ldots, {{\mathbf{w}}_{k{{M}_{\text{RF}}}}}   \right] \in \mathbb{C}^{M\times M_{\text{RF}}} 
	$} is the analog beamformer of the $k$-th subarray,  ${{\mathbf{w}}_{kt}}$, $\forall t=1, 2, \cdots,M_{\text{RF}}$ are the phase shifter values at the $k$-th subarray,  with each  element in ${{\mathbf{w}}_{kt}}$ has unit modulus and continuous phase, i.e., $\left|\mathbf{w}_{k t}[m]\right|^2=1, \forall m\in\mathcal{M}$.
	Additionally, the normalized transmit power constraint is given by $\left\|\mathbf{W}_{\text{RF}} \mathbf{W}_{\text{BB}}\right\|_F^2\le N_s$.

	\section{Performance Metrics and Problem Formulation}\label{section_performetric_problem}
This section first introduces the piecewise-far-field channel model for characterizing the communication and sensing channels.
 Based on this model, we then present the performance metrics for communication and sensing, and formulate the joint analog-digital beamforming optimization problem. 
 
 The small subarray aperture  resulting from $M \ll N$ limits the near-field range of each subarray to only a few meters, making it more likely for the user, target, or scatterer to be located in the far-field region of each subarray while remaining in the near-field region of the overall array.
 The small subarray aperture resulting from $M \ll N$ leads to a negligible near-field range for each subarray. 
  Therefore, we adopt the piecewise-far-field model, which assumes far-field propagation within each subarray and near-field propagation among subarrays \cite{cui2024nearwide,yan2021joint}.
 \begin{remark}
     \emph{\textcolor{black}{For example, considering a modular array with $K=6$, $M=32$ at 38 GHz, the subarray Rayleigh distance is only 2.5 m. In contrast, the overall array has a 40 m Rayleigh distance when collocated (half-wavelength spacing) and 250 m when widely spaced (10-wavelength spacing). Thus, users and targets are typically in the far-field of each subarray but near-field of the overall array}.}
 \end{remark} 
 %


	\subsection{Communication Performance}

  As the size of the user's antenna array is  also considerably small, the communication channel between the transmit subarrays  and  the user's receive array can be modeled under the far-field assumption. 
	Therefore, the far-field communication channel between  the $k$-th transmit subarray and the user is given by\cite{zhu2023subarray,yan2021joint}:
	\begin{equation}\label{subchannel}
       \setlength{\abovedisplayskip}{3pt}
	\setlength{\belowdisplayskip}{3pt}
\mathbf{H}_{\text{sub}}^{k}=\sum\limits_{p=1}^{{{N}_{p}}}{\mu _{p}^{k}}\mathbf{a}_{cp}^{k}\left( \theta _{cp}^{k} \right)\mathbf{a}_{tp}^{k}{{\left( \theta _{tp}^{k} \right)}^{H}},
	\end{equation}
	where  $N_p$ denotes the total number of distinguishable paths between \ac{tx} and the user, including both \ac{los} and \ac{nlos} propagation, ${\mu _{p}^{k}}=\left| \mu _{p}^{k} \right|{{e}^{-j\frac{2\pi }{\lambda }D_{p}^{k}}}$ denotes the complex gain of $p$-th path between the $k$-th transmit subarray and the user's receive array,  which captures the near-field effect across different subarrays as it depends on the exact distance $D_{p}^{k}$ between the $k$-th transmit subarray's reference antenna and the user's receive array's reference antenna along the $p$-th multipath, with $p = 1$ and $p>1$ representing the \ac{los} path and the \ac{nlos} path, respectively.
	$\theta *{cp}^{k}$ and $\theta *{tp}^{k}$ represent the \ac{aoa} relative to the receive array and \ac{aod} relative to the $k$-th transmit subarray along the $p$-th propagation path between the \ac{tx} and user, as illustrated in Fig.~\ref{system_model}.
	The  array steering vector of the $p$-th path for the $k$-th transmit subarray is given by
	\textcolor{black}{$\mathbf{a}_{tp}^{k}\left( \theta _{tp}^{k} \right)={{\left[ 1, {{e}^{-j\frac{2\pi }{\lambda }d\sin \theta _{tp}^{k}}}, \ldots, {{e}^{-j\frac{2\pi }{\lambda }(M-1)d\sin \theta _{tp}^{k}}} \right]}^{T}}$.}
	Similarly, the  array steering vector of the $p$-th path for the user's receive array is given by $\mathbf{a}_{cp}^{k}\left( \theta _{cp}^{k} \right)={{\left[ 1, {{e}^{-j\frac{2\pi }{\lambda }d\sin \theta _{cp}^{k}}}, \ldots, {{e}^{-j\frac{2\pi }{\lambda }({{N}_{c}}-1)d\sin \theta _{cp}^{k}}} \right]}^{T}}$.

Based on the piecewise-far-field channel model, we can characterize the overall communication channel between the user and the transmit array by concatenating the individual far-field subarray channels in \eqref{subchannel}.
The complete mathematical expression is presented in (\ref{hspm_com}) at the bottom of this page.
	\begin{figure*}[!b]
 	\rule[-0pt]{18.5 cm}{0.05em}
		\begin{equation} \label{hspm_com}
			\begin{aligned}
				& {{\mathbf{H}}_{c}}=\left[ \begin{matrix}
					\mathbf{H}_{\text{sub}}^{1} & \cdots  & \mathbf{H}_{\text{sub}}^{K}  \\
				\end{matrix} \right]  
				=\left[ \begin{matrix}
					\sum\limits_{p=1}^{{{N}_{p}}}{\mu _{p}^{1}}\left[ \mathbf{a}_{cp}^{1}\left( \theta _{cp}^{1} \right)\mathbf{a}_{tp}^{1}{{\left( \theta _{tp}^{1} \right)}^{H}} \right] & \cdots  & \sum\limits_{p=1}^{{{N}_{p}}}{\mu _{p}^{K}}\left[ \mathbf{a}_{cp}^{K}\left( \theta _{cp}^{K} \right)\mathbf{a}_{tp}^{K}{{\left( \theta _{tp}^{K} \right)}^{H}} \right]  
				\end{matrix} \right] .
			\end{aligned}
		\end{equation}
	\end{figure*}
 
  In the \ac{mmwave}/sub-THz band, the channels tend to be more sparse due to significant losses caused by large reflection, diffraction, and scattering effects \cite{xing2021milli,gougeon2019ray}.
  As paths with insignificant path gains can be disregarded, the number of multipaths in the \ac{mmwave}/sub-THz band is quite limited. 
  This limitation can decrease the rank of the communication channel, leading to a decline in the communication spectral efficiency.  
To thoroughly understand this rank-dependent behavior, we provide a rigorous analysis of the rank of ${{\mathbf{H}}_{c}}$ in Lemma \ref{Hc_rank} below.
	\begin{lemma} \label{Hc_rank}
		\emph{ The rank of  ${{\mathbf{H}}_{c}}$ satisfies
			\begin{equation}
				\min\{N_p, N_c\}\leq\mathrm{rank}\left( {{\mathbf{H}}_{c}} \right) \leq \min\{KN_p, N_c\}.
			\end{equation}
		}
	\end{lemma} 
	\begin{proof} Please see Appendix \ref{channel_rank_proof}.
	\end{proof}

\begin{remark}
   \emph{\textcolor{black}{ For sufficiently large $N_c$, the rank of traditional far-field channel is limited by the number of multipaths $N_p$ \cite{gao2016energy}. In contrast, Lemma~\ref{Hc_rank} demonstrates that the subarray-based piecewise-far-field channel model can provide a more sufficient rank that is lower-bounded by $N_p$ and can potentially reach $KN_p$, enabling the support of up to $KN_p$ independent data streams.}} 
\end{remark}

	Given the transmit  signal in (\ref{transmit signal}), the received signal of the user is expressed as
	\begin{equation}
		{{\mathbf{Y}}_{c}}={{\mathbf{H}}_{c}}\mathbf{X}+{{\mathbf{Z}}_{c}}={{\mathbf{H}}_{c}}{{\mathbf{W}}_{\text{RF}}}{{\mathbf{W}}_{\text{BB}}}\mathbf{S}+{{\mathbf{Z}}_{c}},
	\end{equation}
	where ${{\mathbf{Z}}_{c}} \in \mathbb{C}^{N_c\times L}$ is a zero-mean complex Gaussian noise matrix with covariance matrix $\sigma_c^2 \mathbf{I}_{N_c}$.
	Then, the achievable communication spectral efficiency  can be calculated as
	\begin{equation} \label{comm_rate}
           \setlength{\abovedisplayskip}{2pt}
	\setlength{\belowdisplayskip}{2pt}
		C=\log \det \left( \mathbf{I}+\frac{1}{\sigma_c^2}{{\mathbf{H}}_{c}}{{\mathbf{W}}_{\text{RF}}}{{\mathbf{W}}_{\text{BB}}}\mathbf{W}_{\text{BB}}^{H}\mathbf{W}_{\text{RF}}^{H}\mathbf{H}_{c}^{H} \right).
	\end{equation}

	\subsection{Sensing Performance}

For radar sensing, we consider a target of interest located at $\mathbf{l}_{0}=(r_0 \sin {\phi}_0, r_0 \cos {\phi}_0)$ and $Q-1$ signal-dependent uncorrelated interferences located at $\mathbf{l}_{q}=(r_q \sin {\phi}_q, r_q \cos {\phi}_q), \forall q \in \{1, 2, \cdots, Q-1\}$. 
Therefore, the received signal at the \ac{bs} over $L$ symbols can be expressed as
	\begin{equation} \label{sensing channel}
		{{\mathbf{Y}}_{s}}=\underbrace{{{\beta }_{{0}}}{{\mathbf{g}}_{r0}}\mathbf{g}_{t0}^{H}\mathbf{X}}_{\text{Target reflection}}+\underbrace{\sum\nolimits_{q=1}^{Q-1}{{{\beta }_{q}}{{\mathbf{g}}_{rq}}\mathbf{g}_{tq}^{H}\mathbf{X}}}_{\text{Echo signal of interferences}}+{{\mathbf{Z}}_{s}}
	\end{equation}
	where $\beta_q$ denotes the complex reflection coefficient proportional to the \ac{rcs} of the $q$-th object, with $\mathbb{E}\left\{\left|\beta_q\right|^2\right\}=\alpha_q^2$ ($q=0$ for the target and $q \geq 1$ for the  interferences), 
	the \ac{awgn} is represented by ${{\mathbf{Z}}_{s}}=\left[ {{\mathbf{z}}_{1}}, {{\mathbf{z}}_{2}}, \ldots ,{{\mathbf{z}}_{L}} \right]\in {{\mathbb{C}}^{N\times L}}$, where each column is an \ac{iid} circularly symmetric complex Gaussian random vector with zero mean and covariance $\mathbf{R}_{s}=\sigma_s^2 \mathbf{I}_{N}$.
	Moreover, based on the piecewise-far-field channel model, the \ac{rx} and \ac{tx} array response vectors for the $q$-th object at the BS are denoted by ${{\mathbf{g}}_{rq}}$ and ${{\mathbf{g}}_{tq}}$, respectively, and can be expressed as
	\begin{subequations}\label{sensing response}
		\begin{align}
			&{{\mathbf{g}}_{tq}}= \left( \mathrm{diag}\left( {{\boldsymbol{\nu }}_{tq}} \right)\otimes {{\mathbf{I}}_{M}} \right){{{\mathbf{\hat{a}}}}_{tq}}, \\ 
			&{{\mathbf{g}}_{rq}}= \left( \mathrm{diag}\left( {\boldsymbol{\nu }_{rq}} \right)\otimes {{\mathbf{I}}_{M}} \right){{{\mathbf{\hat{a}}}}_{rq}}, 
		\end{align}
	\end{subequations}
where
 \begin{equation}  
 \begin{aligned}
&{{\mathbf{\hat{a}}}_{tq}}\!=\!\text{vec}\!\left( \left[ 
   \mathbf{a}_{tq}^{1}, \mathbf{a}_{tq}^{2}, \cdots,  \mathbf{a}_{tq}^{K} \right] \right), \\
& {{\mathbf{\hat{a}}}_{rq}}\!=\!\text{vec}\!\left( \left[ 
   \mathbf{a}_{rq}^{1}, \mathbf{a}_{rq}^{2}, \cdots,  \mathbf{a}_{rq}^{K}\right] \right),
 \end{aligned}
 \end{equation}
where
	$\mathbf{a}_{tq}^{k}={{[ 1, {{e}^{-j\frac{2\pi }{\lambda }d\sin \phi  _{tq}^{k}}}, \ldots, {{e}^{-j\frac{2\pi }{\lambda }( M-1 )d\sin \phi  _{tq}^{k}}}]}^{T}}$  and $\mathbf{a}_{rq}^{k}={{[ 1, {{e}^{-j\frac{2\pi }{\lambda }d\sin \phi _{rq}^{k}}}, \ldots, {{e}^{-j\frac{2\pi }{\lambda }( M-1 )d\sin \phi _{rq}^{k}}}]}^{T}}$ denote the intra-subarray response vectors for the \ac{tx} and \ac{rx}, considering the far-field effect within individual subarrays.
	The angles ${\phi}_{tq}^k$ and ${\phi}_{rq}^k$ denote the direction of the $q$-th object with respect to the positive $y$-axis, as observed from the perspective of the $k$-th subarray at the \ac{tx} and \ac{rx}, respectively, which are given by
	\begin{subequations}
       \setlength{\abovedisplayskip}{2pt}
	\setlength{\belowdisplayskip}{4pt}
		\begin{align}
			& \phi _{tq}^{k}=\text{arc} \sin \left(\frac{r_q\sin \phi_q -{{D}_{0}}-\left( k-1 \right){{d}_{s}}}{\left\| {{\mathbf{l}}_{q}}-\mathbf{l}_{k}^{t} \right\|}\right), \\ 
			&  \phi _{rq}^{k}=\text{arc} \sin \left(\frac{r_q\sin \phi_q +{{D}_{0}}+\left( k-1 \right){{d}_{s}}}{\left\| {{\mathbf{l}}_{q}}-\mathbf{l}_{k}^{r} \right\|}\right).	
		\end{align}
	\end{subequations}
	Additionally, ${{\boldsymbol{\nu}}_{tq}}$ and ${{\boldsymbol{\nu}}_{rq}}$ denote the inter-subarray response vectors of \ac{tx} and \ac{rx} toward the $q$-th object,  characterizing the near-field effect between subarrays, which can be expressed as
	\begin{equation}\label{nu_tq}
           \setlength{\abovedisplayskip}{3pt}
	\setlength{\belowdisplayskip}{3pt}
		\begin{aligned}
			& {\boldsymbol{\nu }}_{tq}=\left[  {e^{-j\frac{2\pi }{\lambda }\left\| {{\mathbf{l}}_{q}}-\mathbf{l}_{1}^{t} \right\|}},\ldots , {{}e^{-j\frac{2\pi }{\lambda }\left\| {{\mathbf{l}}_{q}}-\mathbf{l}_{K}^{t} \right\|}} \right]^T,\\
			& {\boldsymbol{\nu }}_{rq}=\left[  {{e}^{-j\frac{2\pi }{\lambda }\left\| {{\mathbf{l}}_{q}}-\mathbf{l}_{1}^{r} \right\|}},\ldots , {{e}^{-j\frac{2\pi }{\lambda }\left\| {{\mathbf{l}}_{q}}-\mathbf{l}_{K}^{r} \right\|}} \right]^T,	
		\end{aligned}
	\end{equation}
where $\mathbf{l}_{k}^{t}$ and $\mathbf{l}_{k}^{r}$ denote the positions of the reference antennas in the $k$-th subarray at \ac{tx} and \ac{rx}, i.e., $\mathbf{l}_{k}^{t}=\mathbf{l}_{k,1}^{t}$ and $\mathbf{l}_{k}^{r}=\mathbf{l}_{k,1}^{r}$.

 Then, the received signal  is processed by applying the receive beamformer $\mathbf{w}\in \mathbb{C}^{N\times1}$, resulting in the following output at the \ac{rx}:
	\begin{equation}
       \setlength{\abovedisplayskip}{2pt}
	\setlength{\belowdisplayskip}{4pt}
		\begin{aligned}
			{\mathbf{y}_{s}}={{\mathbf{w}}^{H}}{{\beta }_{0}}{{\mathbf{g}}_{r0}}\mathbf{g}_{t0}^{H}\mathbf{X}+{{\mathbf{w}}^{H}}\sum\limits_{q=1}^{Q-1}{{{\beta }_{q}}{{\mathbf{g}}_{rq}}\mathbf{g}_{tq}^{H}\mathbf{X}}+{{\mathbf{w}}^{H}}{{\mathbf{Z}}_{s}}.
		\end{aligned}
	\end{equation}
	Subsequently, the sensing \ac{scnr} can be calculated as 
	\begin{equation} \label{sensing_scnr}
           \setlength{\abovedisplayskip}{3pt}
	\setlength{\belowdisplayskip}{2pt}
		\begin{aligned}
			& {\gamma}_s =\frac{E\left[ {{\left\| {{\mathbf{w}}^{H}}{{\beta }_0}{{\mathbf{g}}_{r0}}\mathbf{g}_{t0}^{H}\mathbf{X} \right\|}^{2}} \right]}{E\left[ {{\left\| {{\mathbf{w}}^{H}}\sum\nolimits_{q=1}^{Q-1}{{{\beta }_{q}}{{\mathbf{g}}_{rq}}\mathbf{g}_{tq}^{H}\mathbf{X}} \right\|}^{2}} \right]+E\left[ {{\left\| {{\mathbf{w}}^{H}}{{\mathbf{Z}}_{s}} \right\|}^{2}} \right]}\\
            & =\frac{{{\mathbf{w}}^{H}}{{\mathbf{G}}_{0}}\left( {{\mathbf{W}}_{\text{RF}}}{{\mathbf{W}}_{\text{BB}}}\mathbf{W}_{\text{BB}}^{H}\mathbf{W}_{\text{RF}}^{H}\right)\mathbf{G}_{0}^{H}\mathbf{w}}{{{\mathbf{w}}^{H}}\left( \sum\nolimits_{q=1}^{Q-1}{{{\mathbf{G}}_{q}}\left( {{\mathbf{W}}_{\text{RF}}}{{\mathbf{W}}_{\text{BB}}}\mathbf{W}_{\text{BB}}^{H}\mathbf{W}_{\text{RF}}^{H}\right)\mathbf{G}_{q}^{H}}+\sigma _{s}^{2}{{\mathbf{I}}_{N}} \right)\mathbf{w}},
		\end{aligned}
	\end{equation}
where ${{\mathbf{G}}_{q}}\triangleq \alpha_{q}{{\mathbf{g}}_{rq}}\mathbf{g}_{tq}^{H}$, $q=0,1,\ldots,Q-1$.

	When the transmit covariance matrix  $\mathbf{R}_X$ is given, the optimal receive beamformer  $\mathbf{w}^{*}$ that maximizes the \ac{scnr} can be obtained by solving the equivalent \ac{mvdr} problem\cite{capon1969high}.
 The closed-form optimal solution is given by 
	\begin{equation}\label{w_calculate}
{\mathbf{w}^{*}}=\frac{{{\left(\boldsymbol{\Sigma} +\sigma _{s}^{2}{{\mathbf{I}}_{N}} \right)}^{-1}}{{\mathbf{g}}_{r0}}}{\mathbf{g}_{r0}^{H}{{\left( \boldsymbol{\Sigma} +\sigma _{s}^{2}{{\mathbf{I}}_{N}} \right)}^{-1}}{{\mathbf{g}}_{r0}}},
	\end{equation}
	where $\boldsymbol{\Sigma} \triangleq \sum\nolimits_{q=1}^{Q-1}{\alpha _{q}^{2}{{\mathbf{g}}_{rq}}\mathbf{g}_{tq}^{H}{{\mathbf{R}}_{X}}{{\mathbf{g}}_{tq}}\mathbf{g}_{rq}^{H}}$.
It is worth noting that once the optimal transmit covariance matrix $\mathbf{R}_{X}$ is determined,  the optimal receive beamformer ${\mathbf{w}^{*}}$ can be readily computed using the closed-form expression in \eqref{w_calculate}. 
  Therefore, in this work, we concentrate on optimizing the transmit beamforming while employing a fixed receive beamformer $\mathbf{w}$ to avoid the more complex joint transmit-receive beamforming design.


	\subsection{Problem Formulation}
    We jointly optimize the transmit and receive beamformers  to maximize the achievable communication spectral efficiency while satisfying the following constraints: (i) the total transmit power budget, (ii) the hardware limitations imposed by the group-connected architecture and constant-modulus phase shifters, and (iii) the minimum required sensing \ac{scnr} for reliable target detection.
	Based on \eqref{cov_matrix_RX}, (\ref{comm_rate}) and (\ref{sensing_scnr}), the optimization problem can be given by
    \begin{figure*}[b!]
        
    \end{figure*}
\begin{subequations}\label{orig_optimization_problem}
		\begin{align}
			\label{original_objective}
			\max \limits_{{{\mathbf{W}}_{\text{BB}}},{{\mathbf{W}}_{\text{RF}}}, \mathbf{w} }\! 
			&\!  \log \det\! \left( \! \mathbf{I}\! +\! \frac{1}{\sigma_c^2}{{\mathbf{H}}_{c}}{{\mathbf{W}}_{\text{RF}}}\! {{\mathbf{W}}_{\text{BB}}}\! \mathbf{W}_{\text{BB}}^{H}\mathbf{W}_{\text{RF}}^{H}\mathbf{H}_{c}^{H}\!  \! \right) \\
			{\rm s.t.}
			\label{original_pt_constraint}
			& \left\| {{\mathbf{W}}_{\text{RF}}}{{\mathbf{W}}_{\text{BB}}} \right\|_{F}^{2}\le{{N}_{s}},  \\
			\label{original_analog_constraint}
			&  {{\mathbf{W}}_{\text{RF}}}\in {{\mathcal{A}}_{F}},  \\
    \label{sensing_scnr_constraint}
			&\frac{{{\mathbf{w}}^{H}}{{\mathbf{G}}_{0}}\left( {{\mathbf{W}}_{\text{RF}}}{{\mathbf{W}}_{\text{BB}}}\mathbf{W}_{\text{BB}}^{H}\mathbf{W}_{\text{RF}}^{H}\right)\mathbf{G}_{0}^{H}\mathbf{w}}{{{\mathbf{w}}^{H}}\left( \sum\nolimits_{q=1}^{Q-1}{{{\mathbf{G}}_{q}}\left( {{\mathbf{W}}_{\text{RF}}}{{\mathbf{W}}_{\text{BB}}}\mathbf{W}_{\text{BB}}^{H}\mathbf{W}_{\text{RF}}^{H}\right)\mathbf{G}_{q}^{H}}+\sigma _{s}^{2}{{\mathbf{I}}_{N}} \right)\mathbf{w}} \geq {{\Gamma }_{s}},
		\end{align}
	\end{subequations}
	where ${{\mathcal{A}}_{F}}$ is the set of block matrices, with each block being an $M \times M_{\text{RF}}$ matrix with constant-magnitude entries.
	However, due to the specific block-diagonal structure
	of the analog beamformer in ${\mathcal{A}}_{F}$, this type of joint analog-digital optimization problem is non-convex and intractable. 

\section{Low-Complexity  Hybrid Beamforming Design} \label{algorithm_design}
In this section, we first analyze the optimal transmit covariance matrix structure and obtain the closed-form analog beamformer.
Building upon this, we propose two efficient algorithms to optimize the digital beamformer: a joint Riemannian-Euclidean gradient descent method and an \ac{sdr}-based randomization approach.

  \subsection{The Optimal Waveform Covariance Matrix} 
   Before gaining an insight into the solution of the joint analog-digital beamforming optimization problem (\ref{orig_optimization_problem}), we first introduce the optimal transmit covariance matrix $\mathbf{R}_X^{*}$ for achieving the maximum communication spectral efficiency under the sensing \ac{scnr} and transmit power constraints:
   \begin{equation} \label{SDP_Rx}
   	\begin{aligned}
   		\mathbf{R}_X^{*}= &\max \limits_{\mathbf{R}_X }
   		 \ \ \  C \\ 
   		{\rm s.t.}
   		\ \  \ &\text{tr}\left({\mathbf{R}_X }\right) \leq{{N}_{s}},  
   		\ \ {{\mathbf{R}}_{X}}\succeq 0,
   		\ \  {{\gamma }_{s}} \geq {{\Gamma }_{s}}.
   	\end{aligned}
   \end{equation}
   Building upon this, the method for designing the optimal analog and digital beamformers is proposed in the subsequent subsection.
   
   To provide an optimal solutions to the problem  (\ref{SDP_Rx}), we propose to exploit the structure of  the communication channel and sensing array response vectors based on  the piecewise-far-field channel model, respectively.
   Namely, we leverage the following observations:
  
   \subsubsection{Structure of the communication channel} \label{comm_channel_structure}
   We begin by performing the \ac{svd} of the communication channel ${{\mathbf{H}}_{c}}={{\mathbf{U}}_{c}}{{\mathbf{\Sigma }}_{c}}{{\mathbf{V}}_{c}}^{H}$, where ${\mathbf{U}}_{c}\in\mathbb{C}^{N_c\times KN_p}$ and ${\mathbf{V}}_{c}\in\mathbb{C}^{N\times KN_p}$ are unitary matrices, and ${{\mathbf{\Sigma }}_{c}}\in\mathbb{C}^{KN_p\times KN_p}$ is a diagonal matrix of singular values arranged in decreasing order.
   The columns of the unitary matrix ${\mathbf{V}}_{c}$ form an orthonormal basis for the ${\mathbf{H}}_{c}$'s row space.
   Besides, according to Lemma \ref{Hc_rank} and its proof, we note that the $KN_p$ linearly independent vectors $\tilde{\mathbf{a}}_{tp}^{k}\left( \theta _{tp}^{k} \right),\forall k, p$ form another minimal basis for the $\mathbf{H}_c$'s row space.
   Therefore, the columns of ${\mathbf{V}}_{c}$ can be written as linear combinations of $\tilde{\mathbf{a}}_{tp}^{k}\left( \theta _{tp}^{k} \right),\forall k, p$, i.e.,
     \begin{equation} \label{Vc_AT}
            \setlength{\abovedisplayskip}{2pt}
	\setlength{\belowdisplayskip}{4pt}
{{\mathbf{V}}_{c}}=\mathbf{\tilde{A}}_c\mathbf{T},
   \end{equation}
   where $\mathbf{T}\in\mathbb{C}^{KN_p\times KN_p}$ represents the linear transformation matrix, and $\mathbf{\tilde{A}}_c\in\mathbb{C}^{N\times KN_p}$ is a matrix formed by combining $\tilde{\mathbf{a}}_{tp}^{k}\left( \theta _{tp}^{k} \right),\forall k, p$ as column vectors.

   \subsubsection{Structure of Sensing Array Response Vectors}	 \label{sensing_structure}
 The subarray response vectors $\mathbf{a}_{tq}^{1}(\mathbf{l}q),\ldots,\mathbf{a}_{tq}^{K}(\mathbf{l}_q)$ for a given sensing object $q$ are linearly independent due to the widely-spaced subarray configuration, which results in different subarrays observing the target from distinct \ac{aod}.
   Therefore, each sensing array response vector ${{\mathbf{g}}_{tq}}$ can be expressed as a fixed linear combination of $K$ linearly independent vectors as
\begin{equation}
       \setlength{\abovedisplayskip}{3pt}
	\setlength{\belowdisplayskip}{1pt}
    \mathbf{\bar{a}}_{tq}^{k}{{\left( \phi _{tq}^{k} \right)}}\triangleq [\  \underbrace{\mathbf{0}, \ldots , \mathbf{0}}_{k-1},\mathbf{a}_{tq}^{k}{{\left( \phi _{tq}^{k} \right)}^{H}}, \underbrace{\mathbf{0}, \ldots , \mathbf{0}}_{K-k} \ ], \in \mathbb{C}^{1\times KM},
\end{equation}
where $k=1, 2, \cdots, K$, $q=0, 1, \cdots, Q-1$, and $\mathbf{0}$ is an all-zero vector of dimension $1\times M$.

  Then, we have the following new representation of the sensing array response vector:
   \begin{equation}\label{gt_Anu}
          \setlength{\abovedisplayskip}{2pt}
	\setlength{\belowdisplayskip}{2pt}
   	\begin{aligned}
   		{{\mathbf{g}}_{tq}}={{\mathbf{\bar{A}}}_{tq}} {\boldsymbol{\nu }}_{tq},\  \forall q,
   	\end{aligned}
   \end{equation}
 where
  \begin{equation}
  	\begin{aligned}
  		{{\mathbf{\bar{A}}}_{tq}}=\left[ \begin{matrix}
  			\mathbf{\bar{a}}_{tq}^{1}\left( \phi _{tq}^{1} \right),& \mathbf{\bar{a}}_{tq}^{2}\left( \phi _{tq}^{2} \right),& \ldots,& \mathbf{\bar{a}}_{tq}^{K}\left( \phi _{tq}^{K} \right)  \\
  		\end{matrix} \right].
  	\end{aligned}
  \end{equation}

	\subsubsection{Joint Representation of Communication Channel and Sensing Array Response Vectors}
	To establish the joint representation of the communication channel and sensing array response vectors, we first construct a set consisting of  communication and sensing subarray response vectors, which can be given by
	\begin{equation}
		\begin{aligned}
			\mathbf{U}=\left[ {{{\mathbf{\bar{A}}}}_{t1}},\ldots ,{{{\mathbf{\bar{A}}}}_{tQ}},{{{\mathbf{\tilde{A}}}}_{c}} \right]\in {{\mathbb{C}}^{N\times K\left( Q+{{N}_{p}} \right)}}.
		\end{aligned}
	\end{equation}
	Then, based on the observations \eqref{Vc_AT} and \eqref{gt_Anu}, we can obtain
	\begin{equation} \label{gt_linearcom}
			 {{\mathbf{g}}_{t}}\left( {{\mathbf{l}}_{q}} \right)=\mathbf{U}\mathbf{\tilde{ {\boldsymbol{\nu }}}}_{tq} , \  \forall q,
	\end{equation}
	\begin{equation} \label{Vc_linearcom}
	\setlength{\belowdisplayskip}{0.5pt}
	   {{\mathbf{V}}_{c}} =\mathbf{U\tilde{T}},
	\end{equation}
	where
	\begin{equation}
		\mathbf{\tilde{ {\boldsymbol{\nu }}}}_{tq} ={{\left[ \mathbf{0}_{K}^{T},\ldots , {\boldsymbol{\nu }}_{tq}^{T},\ldots ,\mathbf{0}_{K}^{T} \right]}^{T}}\in {{\mathbb{C}}^{K\left( Q+{{N}_{p}} \right)\times 1}},
	\end{equation}
	\begin{equation}
		\begin{aligned}
			\mathbf{\tilde{T}}& =\left[ \begin{aligned}
				& {{\mathbf{0}}_{KQ\times K{{N}_{p}}}} \\ 
				& \mathbf{T} \\ 
			\end{aligned} \right]\in {{\mathbb{C}}^{K\left( Q+{{N}_{p}} \right)\times K{{N}_{p}}}}. 
		\end{aligned}
	\end{equation}

   Based on the above observations, we can obtain the structure of the optimal transmit covariance matrix for the problem (\ref{SDP_Rx}), as provided in the theorem below.
	\begin{theorem} \label{optimal_structure_RX}
		\emph{The optimal transmit waveform covariance matrix $\mathbf{R}_X^{*}$  can be written in the form as
		\begin{equation} \label{optimal_RX}		\mathbf{R}_{X}^{*}=\mathbf{\tilde{U}\Lambda }{{\mathbf{\tilde{U}}}^{H}},
		\end{equation}
	where $\mathbf{\Lambda}\in\mathbb{C}^{K(Q+N_p)\times K(Q+N_p)}$ is a positive semi-definite matrix, and $\mathbf{\tilde{U}}$ is a block diagonal matrix obtained from  $\mathbf{{U}}$ by column permutations, i.e.,
	\begin{equation}\label{tilde_U}
		\begin{aligned}
			\mathbf{\tilde{U}}& =\mathbf{UP}=\left[ \begin{matrix}
				{{{\mathbf{\tilde{U}}}}_{1}} & \cdots  & {{{\mathbf{\tilde{U}}}}_{K}}  \\
			\end{matrix} \right] \\ 
			& =\left[ \begin{matrix}
				{{\mathbf{A}}_{11}} & \mathbf{0} & \cdots  & \mathbf{0}  \\
				\mathbf{0} & {{\mathbf{A}}_{22}} & \cdots  & \mathbf{0}  \\
				\vdots  & \vdots  & \ddots  & \vdots   \\
				\mathbf{0} & \mathbf{0} & \cdots  & {{\mathbf{A}}_{KK}}  \\
			\end{matrix} \right], 
		\end{aligned}
	\end{equation}
    where $\mathbf{P}$ is a permutation matrix of size $K(Q+N_p)$, and $\forall k\in \left\{1,\ldots,K\right\}$ we have
    \begin{subequations} \label{tlide_Uk}
           \setlength{\abovedisplayskip}{3pt}
	\setlength{\belowdisplayskip}{3pt}
    	\begin{align}
    		\label{tilde_U_k}
    		{{\mathbf{\tilde{U}}}_{k}}\!&=\!\left[ \mathbf{\bar{a}}_{t0}^{k}\!\left( \phi _{t0}^{k} \right)\!,\ldots ,\mathbf{\bar{a}}_{tQ-1}^{k}\!\left(\! \phi _{tQ-1}^{k}\! \right)\!,\mathbf{\tilde{a}}_{t1}^{k}\!\left( \theta _{t1}^{k} \right)\!,\ldots ,\mathbf{\tilde{a}}_{t{{N}_{p}}}^{k}\!\left(\! \theta _{t{{N}_{p}}}^{k}\! \right) \!\right],\\
    		\label{A_kk_diagonal}
    		{{\mathbf{A}}_{kk}}\!&=\!\left[ \mathbf{a}_{t0}^{k}\!\left( \phi _{t0}^{1} \right)\!,\ldots ,\mathbf{a}_{tQ-1}^{k}\!\left( \phi _{tQ-1}^{1} \right)\!,\mathbf{a}_{t1}^{k}\!\left( \theta _{t1}^{k} \right)\!,\ldots ,\mathbf{a}_{t{{N}_{p}}}^{k}\!\left(\!\theta _{t{{N}_{p}}}^{k}\!\right)\!\right].
    	\end{align}
    \end{subequations}
	      }
	\end{theorem}
	\begin{proof}
		Please see Appendix \ref{optimal_structure_RX_proof}.
	\end{proof}
  
  Theorem \ref{optimal_structure_RX} demonstrates that the  optimal waveform covariance matrix $\mathbf{R}_X^{*}$ belongs to the column space of the block diagonal matrix ${\mathbf{\tilde{U}}}$, which is also the subspace spanned by the communication and sensing subarray response vectors.
  In the following subsection, we exploit the structure of $\mathbf{R}_X^{*}$ to design the optimal analog and digital beamformers for the problem (\ref{orig_optimization_problem}).

 \subsection{ Equivalent Low-dimensional Optimization Problem}\label{equiv_pro}
 
To fully exploit the spatial multiplexing gain, we transmit $N_s=\mathrm{rank}(\mathbf{H}_c)$ independent data streams to the user.
 In addition, in the \ac{mmwave}/sub-THz band, where reflection and scattering losses are significant, the contributions of high-order reflection and scattering paths can be neglected, leading to a small value of $N_p$. 
 Therefore, in scenarios where the number of subarrays and the number of interferences are limited, it is reasonable to consider that the available RF chains are sufficient and exceed $K(N_p+Q)$.
 Consequently, to enable sufficient spatial \ac{dof}s for both communication and sensing, we  activate $M_{\text{RF}}=(N_p+Q)$ \ac{rf} chains per subarray, yielding a total of $N_{\text{RF}}=K(N_p+Q)$ RF chains in the system.

 Based on the above and  Theorem \ref{optimal_structure_RX}, we have the following lemma for obtaining the optimal analog beamformer for problem (\ref{orig_optimization_problem}).
 \begin{lemma} \label{optimal_WBF}
     \emph{ The optimal analog beamformer $\mathbf{W}_{\text{RF}}^*$ of problem 
(\ref{orig_optimization_problem})  can be expressed as
 \begin{equation} \label{w_rf_opt}
   \setlength{\abovedisplayskip}{3pt}
	\setlength{\belowdisplayskip}{3pt}
{\mathbf{W}_{\text{RF}}^{*}}=\mathbf{\tilde{U}}.
 \end{equation}
}
\end{lemma}
 \begin{proof}
 	Note that both communication and sensing subarray response vectors in (\ref{A_kk_diagonal}) are constant-magnitude phase-only vectors.
 	Therefore, ${\mathbf{\tilde{U}}}$ satisfies both the diagonal matrix constraint and the constant modulus constraint., i.e., ${\mathbf{\tilde{U}}}\in {{\mathcal{A}}_{F}}$, which indicates that ${\mathbf{\tilde{U}}}$ can be applied as the analog beamformer.
 	It is evident that matrix $\mathbf{W}_{\text{BB}}\mathbf{W}{{_{\text{BB}}}^{H}}$ is positive semi-definite, hence ${\mathbf{\tilde{U}}}\mathbf{W}_{\text{BB}}\mathbf{W}{{_{\text{BB}}}^{H}}{\mathbf{\tilde{U}}}^H$ conforms to the structure of the optimal waveform covariance matrix defined in (\ref{optimal_RX}).
 	Therefore, ${\mathbf{\tilde{U}}}$ can be considered as the optimal analog beamformer, which completes the proof.
 \end{proof}
  
   By substituting (\ref{w_rf_opt}) into (\ref{comm_rate}) and (\ref{sensing_scnr}), the achievable communication spectral efficiency and sensing \ac{scnr} can be respectively rewritten as 
  \begin{equation}
    \setlength{\abovedisplayskip}{3pt}
	\setlength{\belowdisplayskip}{3pt}
  	C=\log \det \left( \mathbf{I}+\frac{1}{\sigma _{c}^{2}}{{\mathbf{H}}_{c}}\mathbf{\tilde{U}}{{\mathbf{W}}_{\text{BB}}}\mathbf{W}_{\text{BB}}^{H}{{{\mathbf{\tilde{U}}}}^{H}}\mathbf{H}_{c}^{H} \right),
  \end{equation}
 \begin{equation}\label{scnr_change}
	{{\gamma }_{s}}=\frac{\alpha _{1}^{2}\text{tr}\left( {{\mathbf{W}}_{\text{BB}}}\mathbf{W}_{\text{BB}}^{H}{{\mathbf{\Phi }}_{0}} \right)}{\sum\nolimits_{q=1}^{Q}{\alpha _{q}^{2}\text{tr}\left( {{\mathbf{W}}_{\text{BB}}}\mathbf{W}_{\text{BB}}^{H}{{\mathbf{\Phi }}_{q}} \right)}+\sigma _{s}^{2}\mathbf{w}{{\mathbf{w}}^{H}}},
\end{equation}
  where ${{\mathbf{\Phi }}_{q}}\triangleq{{\mathbf{\tilde{U}}}^{H}}{{\mathbf{g}}_{tq}}\mathbf{g}_{rq}^{H}\mathbf{w}{{\mathbf{w}}^{H}}{{\mathbf{g}}_{rq}}\mathbf{g}_{tq}^{H}\mathbf{\tilde{U}}, q\in \left\{0, \ldots,Q-1\right\}$.
 Moreover, due to the block diagonal structure of $\mathbf{{W}}_{\text{RF}}$, each non-zero element of $\mathbf{W}_{\text{RF}}$ is multiplied with the corresponding row in $\mathbf{{W}}_{\text{BB}}$. Hence, the constraint (\ref{original_pt_constraint}) can be simplified as
 \begin{equation}
 	\left\| {{\mathbf{W}}_{\text{RF}}}{{\mathbf{W}}_{\text{BB}}} \right\|_{F}^{\text{2}}\text{=}M\left\| {{\mathbf{W}}_{\text{BB}}} \right\|_{F}^{\text{2}}\le N_s.
 \end{equation}

 Therefore, the complex joint analog-digital beamforming optimization problem (\ref{orig_optimization_problem}) can be equivalently simplified into a low-dimensional digital beamforming optimization problem with the given optimal analog beamformer, i.e.,
\begin{subequations}\label{digital_optimization_problem}
    \setlength{\abovedisplayskip}{5pt}
	\setlength{\belowdisplayskip}{1pt}
 	\begin{align}
		\label{digital_objective}
  \setlength{\abovedisplayskip}{3pt}
	\setlength{\belowdisplayskip}{3pt}
&\underset{{{\mathbf{W}}_{\text{BB}}}}{\mathop{\max }}\,  \ \ \log \det \left( \mathbf{I}+\frac{1}{\sigma _{c}^{2}}\mathbf{W}_{\text{BB}}^{H}{{{\mathbf{\tilde{U}}}}^{H}}\mathbf{H}_{c}^{H}{{\mathbf{H}}_{c}}\mathbf{\tilde{U}}{{\mathbf{W}}_{\text{BB}}} \right)  \\
  \label{cons_power_WBB}
   &\text{s}.\text{t}. \ \ \text{tr}\left( {{\mathbf{W}}_{\text{BB}}}\mathbf{W}_{\text{BB}}^{H} \right)\le \frac{{{N}_{s}}}{M},  \\
   \label{change_scnr_cons}
    &\ \ \ \alpha_{0}^{2}\text{tr}\left( \mathbf{W}_{\text{BB}}^{H}{{\mathbf{\Phi }}_{0}}{{\mathbf{W}}_{\text{BB}}} \right)\text{-}{{\Gamma }_{s}}\sum\limits_{q=1}^{Q-1}{\alpha _{q}^{2}\text{tr}\left( \mathbf{W}_{\text{BB}}^{H}{{\mathbf{\Phi }}_{q}}{{\mathbf{W}}_{\text{BB}}} \right)}\ge {{\Gamma }_{0}},
 	\end{align}
 \end{subequations}
where $\Gamma_0\triangleq{{\Gamma }_{s}}\sigma_{s}^{2}\mathbf{w}{{\mathbf{w}}^{H}}$.
Note that the constraint (\ref{change_scnr_cons}) is an equivalent transformation of the sensing \ac{scnr} constraint in \eqref{sensing_scnr_constraint} based on \eqref{scnr_change}.

Although the digital beamformer optimization problem \eqref{digital_optimization_problem}, has a reduced dimension, it remains non-convex due to the rank constraint imposed by the limited number of data streams.
To tackle this non-convex problem, we propose two distinct algorithms in the following subsections: a manifold optimization method that directly optimizes the digital beamformer on the rank-constrained space and an \ac{sdr}-based method that obtains a near-optimal solution. 
By investigating these two optimization strategies, we provide a comprehensive framework for solving the rank-constrained digital beamformer design problem.

\subsection{Joint Optimization on Riemannian Manifold and Euclidean Space}\label{rieman_algorithm}
In this subsection, we first derive the structure of the optimal solution to problem \eqref{digital_optimization_problem} without relaxing the rank constraint, and obtain the semi-closed-form solution.
Based on this, we then develop a Riemannian joint gradient descent algorithm.

Let $\mathbf{B}\triangleq {{\mathbf{\tilde{U}}}^{H}}\mathbf{H}_{c}^{H}{{\mathbf{H}}_{c}}\mathbf{\tilde{U}}$, and it follows that $\mathbf{B}$ is a Hermitian matrix with $\mathrm{rank}(\mathbf{B})=\min\{N_{\text{RF}}, N_s\} = N_s$.
Performing eigendecomposition on $\mathbf{B}$ and retaining only the non-zero eigenvalues and their corresponding eigenvectors, we obtain $\mathbf{B}={{\mathbf{U}}_{B}}{{\mathbf{\Sigma }}_{B}}\mathbf{U}_{B}^{H}$, where ${{\mathbf{U}}_{B}}\in\mathbb{C}^{N_{\text{RF}}\times N_s}$, and $\mathbf{\Sigma }_{B}\in\mathbb{C}^{N_s\times N_s}$.
Then, we can obtain the structure of the  optimal solution to the problem \eqref{digital_optimization_problem} in the following Lemma.
\begin{lemma} \label{structure_beamformer}
	\emph{ The optimal solution to the problem \eqref{digital_optimization_problem} is given as
 \begin{equation}
     \mathbf{W}_{\text{BB}}^{*}\text{=}{{\mathbf{U}}_{\text{B}}}\mathbf{\Sigma }_{\text{B}}^{-\frac{1}{2}}\mathbf{U}_{\text{B}}^{H}\mathbf{\tilde{V}\tilde{\Sigma }},
 \end{equation}
    where $\mathbf{\tilde{V}}\in\mathbb{C}^{N_{\text{RF}}\times N_{\text{RF}}}$ is an unitary matrix, and $\tilde{\Sigma }$ is a $N_{\text{RF}}\times N_s$ rectangular diagonal matrix which is defined as
    \begin{equation}
        \mathbf{\tilde{\Sigma }}\text{=}\left[ \begin{matrix}
   \mathrm{diag}\left( \mathbf{b} \right)  \\
   {\mathbf{0}}_{{N_{\mathrm{RF}}}-N_s,N_s}  \\
\end{matrix} \right],
    \end{equation}
    with $\boldsymbol{b} = [b_1, b_2, \ldots, b_{N_s}]^T$.
    }	
\end{lemma}

\begin{proof}
 Please see Appendix \ref{optimal_structure_WBB_proof}.
\end{proof}
It can be observed that the unitary matrix $\mathbf{\tilde{V}}$ forms a complex Stiefel mainifold $\mathcal{M}_s = \{\mathbf{\tilde{V}}\in\mathbb{C}^{N_{\text{RF}}\times N_{\text{RF}}}: \mathbf{\tilde{V}}^H\mathbf{\tilde{V}}=\mathbf{I}_{N_{\text{RF}}}\}$.
Therefore, the problem \eqref{digital_optimization_problem} can be rewritten as 
\begin{subequations}\label{mainfold_opt_problem}
    \setlength{\abovedisplayskip}{4pt}
	\setlength{\belowdisplayskip}{1pt}
    \begin{align}
    \label{object_func}
    \underset{\mathbf{\tilde{V}},\mathbf{b}}{\mathop{\max }}\, & \ \ \log \det \left( \mathbf{I}+\mathbf{\tilde{\Sigma }}{{{\mathbf{\tilde{\Sigma }}}}^{H}} \right)  \\[-3mm]
        \label{constraint_power}
   \text{s}.\text{t}. & \ \ \text{tr}\left( {{{\mathbf{\tilde{V}}}}^{H}}\mathbf{\tilde{B}\tilde{V}\tilde{\Sigma }}{{{\mathbf{\tilde{\Sigma }}}}^{H}} \right)\le \frac{{{N}_{s}}}{M},  \\
    \label{constraint_scnr}
   {} & \ \ \text{tr}\left( {{{\mathbf{\tilde{V}}}}^{H}}\mathbf{\tilde{\Phi }\tilde{V}\tilde{\Sigma }}{{{\mathbf{\tilde{\Sigma }}}}^{H}} \right)\ge {{\Gamma }_{s}}\sigma _{s}^{2}\mathbf{w}{{\mathbf{w}}^{H}},  \\[-1mm]
   {} & \ \ \mathbf{\tilde{V}} \in \mathcal{M}_s,\\[-1mm]
   {} &\ \  \mathbf{b} \in \mathbb{R}^r,
    \end{align}
\end{subequations}
where
\begin{equation}
    \mathbf{\tilde{B}} = {{\mathbf{U}}_{\text{B}}}\mathbf{\Sigma }_{\text{B}}^{-1}\mathbf{U}_{\text{B}}^{H},
\end{equation}
\begin{equation}
    \mathbf{\tilde{\Phi }}\triangleq{{\mathbf{U}}_{\text{B}}}\mathbf{\Sigma }_{\text{B}}^{-\frac{1}{2}}\mathbf{U}_{\text{B}}^{H}\left( \alpha _{0}^{2}{{\mathbf{\Phi }}_{0}}-{{\Gamma }_{s}}\sum\nolimits_{q=1}^{Q-1}{\alpha _{q}^{2}{{\mathbf{\Phi }}_{q}}} \right){{\mathbf{U}}_{\text{B}}}\mathbf{\Sigma }_{\text{B}}^{-\frac{1}{2}}\mathbf{U}_{\text{B}}^{H}.
\end{equation}

We then  use the barrier method to make the inequality constraints \eqref{constraint_power} and \eqref{constraint_scnr}  implicit
in the objective function \eqref{object_func}.
Thus, we have
\begin{equation}
    \setlength{\abovedisplayskip}{4pt}
	\setlength{\belowdisplayskip}{2pt}
\begin{aligned}
    f\left( \mathbf{\tilde{V}}, \mathbf{b} \right) &=-\log \det \left( \mathbf{I}+\mathbf{\tilde{\Sigma }}{{{\mathbf{\tilde{\Sigma }}}}^{H}} \right)\\
    &+\phi \left(\frac{{{N}_{s}}}{M} - \text{tr}\left( {{{\mathbf{\tilde{V}}}}^{H}}\mathbf{\tilde{B}\tilde{V}\tilde{\Sigma }}{{{\mathbf{\tilde{\Sigma }}}}^{H}} \right)\right)\\
    &+\phi \left( \text{tr}\left( {{{\mathbf{\tilde{V}}}}^{H}}\mathbf{\tilde{\Phi }\tilde{V}\tilde{\Sigma }}{{{\mathbf{\tilde{\Sigma }}}}^{H}} \right)-\Gamma_0 \right),
\end{aligned}
\end{equation}
where  $\phi(u)$ is the logarithmic barrier function, i.e.,
\begin{equation}
    \phi(u)= \begin{cases}-\frac{1}{t} \ln (u), & u>0 \\ +\infty, & u \leq 0\end{cases},
\end{equation}
with $t$ barrier parameter $t > 0$.

Consequently,  problem \eqref{mainfold_opt_problem} can be simplified to an unconstrained optimization problem shown  as  below:
\begin{subequations} \label{uncons_problem}
    \setlength{\abovedisplayskip}{4pt}
	\setlength{\belowdisplayskip}{3pt}
    \begin{align}
    \underset{\mathbf{\tilde{V}},\mathbf{b} }{\mathop{\min }}\, & \ \ f\left( \mathbf{\tilde{V}},\mathbf{b} \right)  \\
   \text{s}.\text{t}. & \ \ \text{ }\mathbf{\tilde{V}}\in {{\mathcal{M}}_{s}},\ \  \mathbf{b} \in \mathbb{R}^{N_s}.
    \end{align}
\end{subequations}
 It can be observed that the variables $\mathbf{\tilde{V}}$ and $\mathbf{b}$ are coupled in the objective function. 
 To account for the interaction between these variables, we engage in the simultaneous optimization of both  $\mathbf{\tilde{V}}$ and $\mathbf{b}$, enabling coordinated updates to enhance convergence efficiency along a more effective path.

 As a first step, we derive the  Euclidean  gradients of the objective function with respect to $\mathbf{\tilde{V}}$ and $\mathbf{b}$, respectively.
The gradient of the objective function $f(\mathbf{\tilde{V}}, \mathbf{b})$ with respect to $\mathbf{b}$  is provided in \eqref{gradient_b} at the top of this page, where $(\cdot)_{1:N_s,1:N_s}$ represents the top-left $N_s \times N_s$ block of a matrix.
 Then, the update of $\mathbf{b}$ at the (n)-th iteration on the Euclidean space is
 \begin{equation}
     \mathbf{b}^{(n+1)} := \mathbf{b}^{(n)}+\delta_{\mathbf{b}}^{(n)} \nabla_{\mathbf{b}}^{(n)} f,
 \end{equation}
where the step size $\delta_{\mathbf{b}}^{(n)}$ is determined by line search algorithms, such as the Armijo rule.

\begin{figure*}[!t] 
    \begin{equation}\label{gradient_b}
    \nabla_{\mathbf{b}} f = -2\left( \left( \mathbf{I}+\mathbf{{\tilde{\Sigma }}}\mathbf{{{{\tilde{\Sigma }}}}^{H}} \right)^{-1}_{1:N_s,1:N_s} \right)\mathbf{b}+
      \frac{2}{t}\left( \frac{\mathrm{diag}{{\left( {{{\mathbf{\tilde{V}}}}^{H}}\mathbf{\tilde{B}\tilde{V}} \right)}_{1:{{N}_{s}}}}}{\left( \frac{{{N}_{s}}}{M}-\text{tr}\left( {{{\mathbf{\tilde{V}}}}^{H}}\mathbf{\tilde{B}\tilde{V}\tilde{\Sigma }}{{{\mathbf{\tilde{\Sigma }}}}^{H}} \right) \right)}-\frac{\mathrm{diag}{{\left( {{{\mathbf{\tilde{V}}}}^{H}}\mathbf{\tilde{\Phi }\tilde{V}} \right)}_{1:{{N}_{s}}}}}{\left( \text{tr}\left( {{{\mathbf{\tilde{V}}}}^{H}}\mathbf{\tilde{\Phi }\tilde{V}\tilde{\Sigma }}{{{\mathbf{\tilde{\Sigma }}}}^{H}} \right)-{{\Gamma }_{0}} \right)} \right)\mathrm{diag}(\mathbf{b}).
    \end{equation}
     \rule[-0pt]{18.5 cm}{0.05em}
\end{figure*}

The Euclidean gradient of objective function w.r.t. $\mathbf{\tilde{V}}$ is obtained by
\begin{equation}
    \begin{aligned}
  {{\nabla }_{{\mathbf{\tilde{V}}}}}f& =\frac{2}{t}\left( \frac{1}{-\text{tr}\left( {{{\mathbf{\tilde{V}}}}^{H}}\mathbf{\tilde{B}\tilde{V}\tilde{\Sigma }}{{{\mathbf{\tilde{\Sigma }}}}^{H}} \right)+\frac{{{N}_{s}}}{M}}\mathbf{\tilde{B}\tilde{V}\tilde{\Sigma }}{{{\mathbf{\tilde{\Sigma }}}}^{H}} \right. \\ 
  & -\left. \frac{\mathbf{\tilde{\Phi }\tilde{V}\tilde{\Sigma }}{{{\mathbf{\tilde{\Sigma }}}}^{H}}}{\text{tr}\left( {{{\mathbf{\tilde{V}}}}^{H}}\mathbf{\tilde{\Phi }\tilde{V}\tilde{\Sigma }}{{{\mathbf{\tilde{\Sigma }}}}^{H}} \right)-{{\Gamma }_{0}}} \right). 
\end{aligned}
\end{equation}
The tangent space for the complex Stiefel manifold is given by
\begin{equation}
    T_{\mathbf{\tilde{V}}} \mathcal{M}_s= \{\mathbf{Z} \in \mathbb{C}^{N_{\text{RF}} \times N_{\text{RF}}}: \mathbf{Z}^H {\mathbf{\tilde{V}}} + {\mathbf{\tilde{V}}}^H \mathbf{Z} = 0\}.
\end{equation}

For each point $\mathbf{\tilde{V}}\in\mathcal{M}_s$, the decent direction $\boldsymbol{\Delta}_{\mathbf{\tilde{V}}}$ is defined as the projection of the Euclidean gradient ${{\nabla }_{{\mathbf{\tilde{V}}}}}f$ onto the tangent space, which is  obtained by 
\begin{equation} \label{discent_gradient}
    \begin{aligned}
       \boldsymbol{\Delta}_{\mathbf{\tilde{V}}} = \mathbf{\tilde{V}}^H{{\nabla }_{{\mathbf{\tilde{V}}}}}f\mathbf{\tilde{V}}-{{\nabla }_{{\mathbf{\tilde{V}}}}}f.
    \end{aligned}
\end{equation}
Thus, the update of $\mathbf{\tilde{V}}$ at the (n)-th iteration on the tangent space can be given by
\begin{equation}
\mathbf{\bar{V}}^{(n+1)}=\mathbf{\tilde{V}}^{(n)}+{\delta}_{\mathbf{\tilde{V}}}^{(n)} \boldsymbol{\Delta}_{\mathbf{\tilde{V}}}^{(n)},
\end{equation}
where  ${\delta}_{\mathbf{\tilde{V}}}^{(n)}$ is the step size.
However, the new updated point $\mathbf{\bar{V}}^{(n+1)}$ may not necessarily lie on the manifold, necessitating the projection onto the Stiefel manifold.
\begin{proposition}
    \emph{Let $\mathbf{Z} \in \mathbb{C}^{N_{\text{RF}} \times N_{\text{RF}}}$ be a arbitrary matrix. The projection $\mathcal{P_{M_s}}(\mathbf{Z})$ onto the Stiefel manifold is  
    \begin{equation}
        \mathcal{P}_{\mathcal{M}_s}(\mathbf{Z})=\arg \min _{\mathbf{Q} \in \mathcal{M}_s}\|\mathbf{Z}-\mathbf{Q}\|^2.
    \end{equation}
    Additionally, if the \ac{svd} of $\mathbf{Z}$ is $\mathbf{Z}=\mathbf{U}_Z\mathbf{\Sigma}_{Z}\mathbf{V}_Z^H$, then $\mathcal{P}_{\mathcal{M}_s}(\mathbf{Z}) = \mathbf{U}_Z\mathbf{V}_Z^H$.
    }
\end{proposition} 
\begin{proof}
	Please refer to \cite[Proposition 7]{manton2002opt}.
\end{proof}
Therefore, at the (n)-th iteration on the mainfold ${\mathcal{M}_s}$ is given by
\begin{equation}
     \mathcal{P}_{\mathcal{M}_s}(\mathbf{\bar{V}}^{(n+1)})= \mathcal{P}_{\mathcal{M}_s}(\mathbf{\tilde{V}}^{(n)}+{\delta}^{(n)} \boldsymbol{\Delta}_{\mathbf{\tilde{V}}^{(n)}}).
\end{equation}

According to the above discuss, the main procedures of the Riemannian projected steepest descent algorithm for solving problem \eqref{uncons_problem} over $\mathbf{\tilde{V}}$ are described in Algorithm~\ref{algorithm_remannian_modified}.
Upon termination, the algorithm outputs the obtained solution $(\mathbf{\tilde{V}}^{*}, \mathbf{b}^{*})$. 
It is noted that due to the non-convexity of the  problem \eqref{uncons_problem}, the algorithm converges to a local optimum rather than a global one.
Therefore, the  local optimal solution to problem \eqref{digital_optimization_problem} is 
\begin{equation}
    \mathbf{W}_{\text{BB}}^{*}={{\mathbf{U}}_{\text{B}}}\mathbf{\Sigma }_{\text{B}}^{-\frac{1}{2}}\mathbf{U}_{\text{B}}^{H}\mathbf{\tilde{V}^*}\mathbf{\tilde{\Sigma}^*},
\end{equation}
with
\begin{equation}
        \mathbf{\tilde{\Sigma}^*}=\left[ \begin{matrix}
   \mathrm{diag}\left( \mathbf{b}^* \right)  \\
   {{\mathbf{0}}_{N_{\text{RF}}-N_s,N_s}}  \\
\end{matrix} \right].
    \end{equation}

The computational complexity of Algorithm \ref{algorithm_remannian_modified}, hereafter referred to as Riemannian-Euclidean Joint Gradient Descent (RM-JGD), is dominated by the \ac{svd} decomposition of $\mathbf{B} \triangleq \tilde{\mathbf{U}}^H \mathbf{H}_c^H \mathbf{H}_c \tilde{\mathbf{U}}$, the calculation of Riemannian and Euclidean gradients, and the Riemannian projection in each iteration.
Computing $\mathbf{B}$ has a complexity of $\mathcal{O}(N_{\text{RF}}N_cN)$, while the \ac{svd} of $\mathbf{B}$ requires $\mathcal{O}(N_{\text{RF}}^3)$ operations \cite{golub2013matrix}.
The gradient calculations involve matrix multiplications and inversions with a total complexity of $\mathcal{O}(N_{\text{RF}}^3 + N_{\text{RF}}^2N_s + N_{\text{RF}}{N_s}^2)$.
The Riemannian gradient descent update and projection require $\mathcal{O}(N_{\text{RF}}^2)$ and $\mathcal{O}(N_{\text{RF}}^3)$ operations, respectively \cite{edelman1998geometry}.
Thus, the overall complexity of Algorithm \ref{algorithm_remannian_modified} is $\mathcal{O}(I(N_{\text{RF}}N_cN + N_{\text{RF}}^3 + N_{\text{RF}}^2N_s + N_{\text{RF}}{N_s}^2))$, where $I$ is the number of iterations.

\begin{algorithm}[t]
\caption{RM-JGD Algorithm}
\label{algorithm_remannian_modified}
\begin{algorithmic}[1]
\renewcommand{\algorithmicrequire}{\textbf{Initialize}}
\renewcommand{\algorithmicensure}{\textbf{Output}}
\STATE \textbf{Initialize} $N_s$, $M$, ${\Gamma}_{s}$, $t>0$,  the tolerances $\epsilon_{\tilde{\mathbf{V}}} > 0$ and $\epsilon_\mathbf{b} > 0$, the maximum number of iterations $I_{\max}$, and choose a feasible $\mathbf{\tilde{V}}^{(0)}\in\mathbb{C}^{N_{\text{RF}} \times N_{\text{RF}}}$ such that $(\mathbf{\tilde{V}}^{(0)})^H\mathbf{\tilde{V}}^{(0)}=\mathbf{I}$ and $\mathbf{b}^{(0)} \in \mathbb{R}^r$.
\STATE Set $n:=0$,  compute the Riemannian gradient $\xi_{\mathbf{\tilde{V}}}^{(0)}=\mathbf{\Delta}_{\mathbf{\tilde{V}}}(\mathbf{\tilde{V}}^{(0)}, \mathbf{b}^{(0)})$ and the Euclidean gradient $\xi_{\mathbf{b}}^{(0)} = \nabla_{\mathbf{b}} f(\mathbf{\tilde{V}}^{(0)}, \mathbf{b}^{(0)})$;
\WHILE{$n\le I_{\max}$ and $\|\xi_{\mathbf{\tilde{V}}}^{(n)}\|_F^2 \geq \epsilon_{\tilde{\mathbf{V}}}$ or $\|\xi_{\mathbf{b}}^{(n)}\|^2\geq\epsilon_{\mathbf{b}}$}
\STATE Choose the stepsizes $\delta_{\mathbf{\tilde{V}}}^{(n)}$ and $\delta_{\mathbf{b}}^{(n)}$ using backtracking line search such that:
$f\left(\mathcal{P}_{\mathcal{M}_s}(\mathbf{\tilde{V}}^{(n)}+\delta_{\mathbf{\tilde{V}}}^{(n)} \xi_{\mathbf{\tilde{V}}}^{(n)}), \mathbf{b}^{(n)}+\delta_{\mathbf{b}}^{(n)} \xi_{\mathbf{b}}^{(n)}\right) < f(\mathbf{\tilde{V}}^{(n)}, \mathbf{b}^{(n)})$
\STATE Update $\mathbf{\tilde{V}}^{(n+1)} := \mathcal{P}_{\mathcal{M}_s}(\mathbf{\tilde{V}}^{(n)}+{\delta_{\mathbf{\tilde{V}}}}^{(n)} \xi_{\mathbf{\tilde{V}}}^{(n)})$;
\STATE Update $\mathbf{b}^{(n+1)} := \mathbf{b}^{(n)}+{\delta_{\mathbf{b}}}^{(n)} \xi_{\mathbf{b}}^{(n)}$;
\STATE Update $n:=n+1$;
\STATE Compute the descent direction $\xi_{\mathbf{\tilde{V}}}^{(n)}$ as $\xi_{\mathbf{\tilde{V}}}^{(n)} = \mathbf{\Delta}_{\mathbf{\tilde{V}}}(\mathbf{\tilde{V}}^{(n)}, \mathbf{b}^{(n)})$ according to \eqref{discent_gradient};
\STATE Compute the descent direction $\xi_{\mathbf{b}}^{(n)}$ as $\xi_{\mathbf{b}}^{(n)} = \nabla_{\mathbf{b}} f(\mathbf{\tilde{V}}^{(n)}, \mathbf{b}^{(n)})$ according to \eqref{gradient_b};
\ENDWHILE
\STATE \textbf{Output} $\mathbf{\tilde{V}}^{*}={\mathbf{\tilde{V}}}^{(n)}$ and $\mathbf{b}^{*}={\mathbf{b}}^{(n)}$;
\end{algorithmic}
\end{algorithm}

\subsection{SDR-based Randomization for Near-Optimal Solution}
Although the manifold optimization algorithm introduced in the previous subsection offers an efficient approach to directly optimize the digital beamformer on the rank-constrained space, it may converge to a local optimum that is suboptimal compared to the global solution. This limitation arises from the non-convex nature of the problem, where the algorithm's performance heavily relies on the choice of initialization point. 
In contrast, we propose a two-stage approach in this subsection, which first finds a global optimum of the relaxed problem and then obtains a near-optimal solution that satisfies the rank constraint, aiming to find a high-quality solution to the rank-constrained digital beamformer optimization problem while maintaining low computational complexity. 

To tackle the non-convex optimization problem in \eqref{digital_optimization_problem}, a common approach is to apply the \ac{sdr} technique, which relaxes the non-convex rank constraint by introducing a new variable $\mathbf{R}_{\text{BB}} \triangleq \mathbf{W}_{\text{BB}} \mathbf{W}_{\text{BB}}^H$ and dropping the rank constraint $\mathrm{rank}(\mathbf{R}_{\text{BB}}) = N_s$ 
 Consequently, the original problem \eqref{digital_optimization_problem} is transformed into a \ac{sdp} problem as follows:
\begin{subequations}\label{sdr_optimization_problem}
	\begin{align}
		\label{sdp_objective}
		\max \limits_{{\mathbf{R}}_{\text{BB}}}
		& \ \log \det \left( \mathbf{I}+\frac{1}{\sigma _{c}^{2}}{{\mathbf{H}}_{c}}\mathbf{\tilde{U}}{{\mathbf{R}}_{\text{BB}}}{{{\mathbf{\tilde{U}}}}^{H}}\mathbf{H}_{c}^{H} \right) \\\
		{\rm s.t.}
		\label{sdp_pt_constraint}
		&\ \text{tr}\left( {{\mathbf{R}}_{\text{BB}}} \right)\leq\frac{{{N}_{s}}}{M},  \\\
		\label{sdp_sensing_scnr_constraint}
		& \ \alpha _{1}^{2}\text{tr}\left( {{\mathbf{R}}_{\text{BB}}}{{\mathbf{\Phi }}_{0}} \right)- {{\Gamma }_{s}}\sum\limits_{q=1}^{Q-1}{\alpha _{q}^{2}\text{tr}\left( {{\mathbf{R}}_{\text{BB}}}{{\mathbf{\Phi }}_{q}} \right)} \ge {{\Gamma }_{0}}, \\
		\label{sd_constraint}
		& \ {{\mathbf{R}}_{\text{BB}}}\succeq 0,
	\end{align}
  \end{subequations}
where  ${{\mathbf{R}}_{\text{BB}}}\triangleq {{\mathbf{W}}_{\text{BB}}}\mathbf{W}_{\text{BB}}^{H}$, and  the rank constraint $\mathrm{rank}({\mathbf{R}}_{\text{BB}}) = N_s$ is neglected.

 The relaxed problem (\ref{sdr_optimization_problem}) is convex and can  be efficiently solved using interior-point method.
 However,  the optimal solution $\mathbf{R}_{\text{BB}}^*$ to the relaxed problem may not satisfy the rank constraint $\mathrm{rank}(\mathbf{R}_{\text{BB}}) = N_s$.
To recover a rank-constrained solution to the original problem \eqref{digital_optimization_problem}, we apply the randomization technique, which generates a set of candidate solutions from $\mathbf{R}_{\text{BB}}^*$ and selects the one that maximizes the objective function \eqref{sdp_objective} while satisfying the constraints \eqref{sdp_pt_constraint} and \eqref{sdp_sensing_scnr_constraint}.
The detailed procedure of the proposed algorithm, referred to as SDR-based Randomization for Rank-constrained Solution (SDR-RRS), is summarized in Algorithm \ref{alg:lc_sdr_rrs}. 
\begin{algorithm}[t]
\caption{SDR-RRS Algorithm}
\label{alg:lc_sdr_rrs}
\begin{algorithmic}[1]
\renewcommand{\algorithmicrequire}{\textbf{Input:}}
\renewcommand{\algorithmicensure}{\textbf{Output:}}
\STATE \textbf{Initialize} $\mathbf{H}_c$, $\mathbf{\tilde{U}}$, $\sigma_c^2$, $N_s$, $M$, $\Gamma_s$, $\alpha_q$, $\mathbf{\Phi}_q$, $\sigma_s^2$, $\mathbf{w}$.
\STATE Solve the SDP problem \eqref{sdr_optimization_problem} using interior-point method to obtain $\mathbf{R}_{\text{BB}}^*$;
\STATE Compute the eigenvalue decomposition of $\mathbf{R}_{\text{BB}}^* = \mathbf{U} \mathbf{\Lambda} \mathbf{U}^H$;
\STATE Let $\mathbf{U}_{N_s}$ be the matrix containing the eigenvectors corresponding to the $N_s$ largest eigenvalues;
\STATE Set $\mathbf{V} = \mathbf{U}_{N_s} \mathbf{\Lambda}_{N_s}^{1/2}$, where $\mathbf{\Lambda}_{N_s}$ is the diagonal matrix containing the $N_s$ largest eigenvalues;
\FOR{$i = 1, \dots, 10N_s$}
    \STATE Generate a random matrix $\mathbf{Z}_i \in \mathbb{C}^{N_s \times N_s}$ with i.i.d. entries drawn  from $\mathcal{CN}(0, 1)$;
    \STATE Set $\mathbf{W}_i = \mathbf{V} \mathbf{Z}_i$;
    \STATE Scale $\mathbf{W}_i$ to satisfy the power constraint \eqref{sdp_pt_constraint}:\\
        $\mathbf{W}_i \leftarrow \sqrt{\frac{N_s}{M \cdot \textit{tr}(\mathbf{W}_i \mathbf{W}_i^H)}} \mathbf{W}_i$;
\ENDFOR
\STATE Choose the best solution among $\{\mathbf{W}_i\}_{i=1}^{10N_s}$ that satisfies the SCNR constraint \eqref{sdp_sensing_scnr_constraint} and maximizes the objective \eqref{digital_objective}:
 \begin{align*}
        \mathbf{W}^* \!=\! \arg\max_{\mathbf{W}_i} \left\{ \!\log \det \left( \!\mathbf{I} \!+\! \frac{1}{\sigma_c^2} \mathbf{W}_i^H \mathbf{\tilde{U}}^H \mathbf{H}_c^H \mathbf{H}_c \mathbf{\tilde{U}} \mathbf{W}_i\! \right) \!\right\}
 \end{align*}
subject to \eqref{sdp_sensing_scnr_constraint};
\STATE \textbf{Output} $\mathbf{W}_{\text{BB}} = \mathbf{W}^*$
\end{algorithmic}
\end{algorithm}

The computational complexity of Algorithm \ref{alg:lc_sdr_rrs} is dominated by solving the \ac{sdp} problem in Step 1 and the eigenvalue decomposition in Step 2.
The interior-point method for solving the \ac{sdp} problem has a worst-case complexity of $\mathcal{O}(N^6_{\text{RF}})$ \cite{fujisawa1997}. 
The eigenvalue decomposition of an $N_{\text{RF}} \times N_{\text{RF}}$ matrix has a complexity of $\mathcal{O}(N^3_{\text{RF}})$ \cite{golub2013matrix}. 
The remaining steps involve matrix multiplications and scaling operations, which have a complexity of $\mathcal{O}(N^2_{\text{RF}}N_s)$. 
Therefore, the overall computational complexity of Algorithm \ref{alg:lc_sdr_rrs} is $\mathcal{O}(N^6_{\text{RF}}+N^3_{\text{RF}}+IN^2_{\text{RF}}N_s)$, where $I$ is the number of randomization iterations.

 Compared to the RM-JGD algorithm, the proposed SDR-RRS algorithm has several advantages.
 First, by relaxing the non-convex rank constraint, the SDR-RRS transforms the original problem into a convex \ac{sdp} problem, which can be globally solved in polynomial time.
Second, the randomization technique in SDR-RRS generates multiple candidate solutions based on the global optimum of the relaxed problem, ensuring better solution than RM-JGD, where the obtained local optimum highly depends on random initialization and may be far from the global optimum. 

\section{Simmulation Results}\label{section_simulation}

In this section, we evaluate the performance of the proposed modular \ac{xl-mimo} \ac{isac} system and algorithms through extensive simulations.
Unless otherwise specified, the BS is equipped with a modular XL-array of $K=6$ subarrays, each containing $M=32$ antennas. The operating frequency is 38 GHz with half-wavelength antenna spacing $d=\lambda/2\approx 0.00395$ m.
The user is equipped with $N_c = 16$ antennas and located at a distance of 40 m and an angle of $15^{\circ}$ relative to the \ac{tx} center.
The channel consists of one \ac{los} and three \ac{nlos} paths, with scatterers randomly distributed between 5-30 m and $-60^{\circ}$ to $60^{\circ}$. The path gains follow the 3GPP TR 38.901 specification \cite{3gpp2020study}.
The target is positioned at 30 m and $30^{\circ}$, with two interferences at the same range but at angles of $40^{\circ}$ and $-30^{\circ}$. 
The noise power for communication and sensing are set to -30 dBm and -20 dBm, respectively.
For simplicity, we adopt the omnidirectional transmission with $\mathbf{R}_X=\mathbf{I}$  to calculate the fixed receive beamformer $\mathbf{w}$ according to \eqref{w_calculate} in the simulations.
 Once the optimal transmit covariance matrix $\mathbf{R}_X^*$ is obtained, we substitute it into \eqref{w_calculate} to derive the corresponding optimal receive beamformer $\mathbf{w}^*$.

\begin{figure} [t]
\setlength{\abovecaptionskip}{-0.1cm}
		\centering		\includegraphics[width=0.35\textwidth]{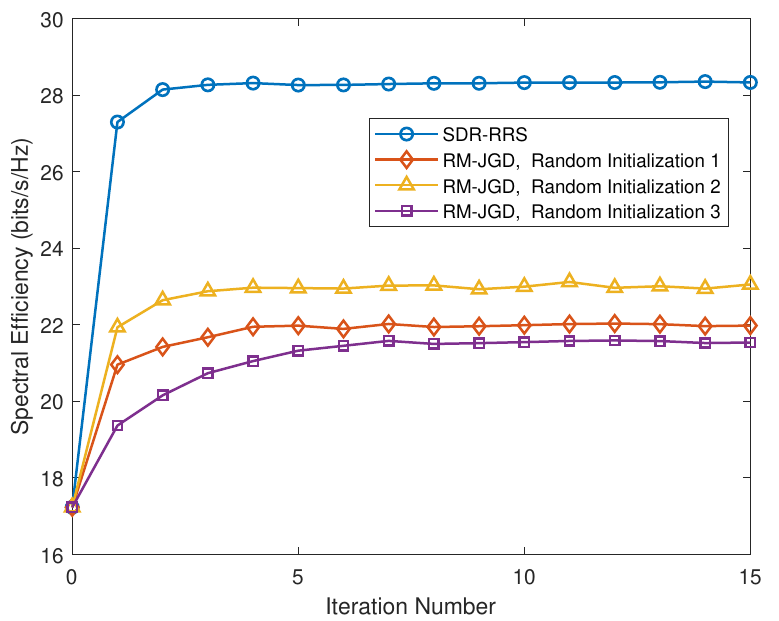}
		\caption{Influence of different barrier parameters and initializations on convergence of Algorithm \ref{algorithm_remannian_modified} .}
\label{overallalg_convergence}\vspace{-3mm}
  \end{figure}

\begin{figure} [t]
\setlength{\abovecaptionskip}{-0.1cm}
		\centering	\includegraphics[width=0.35\textwidth]{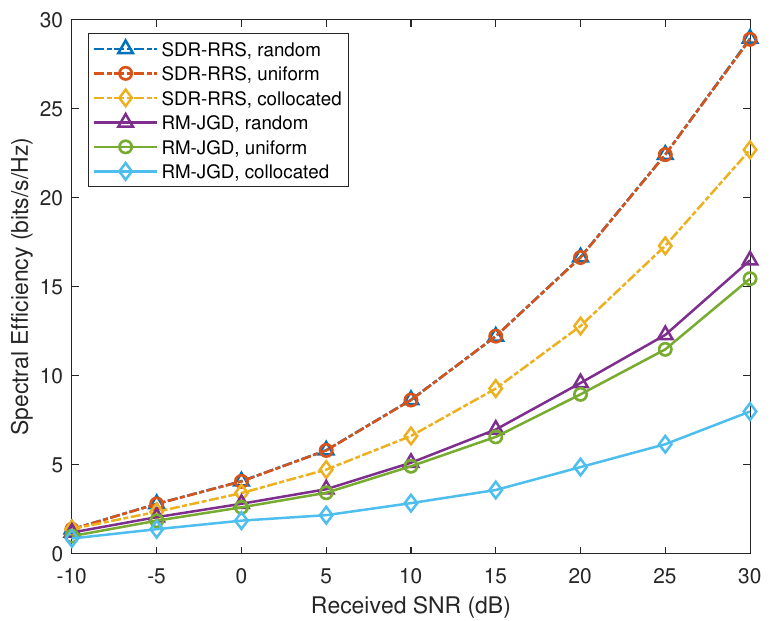}
		\caption{Communication spectral efficiency versus received SNR for different subarray layouts.}
\label{Comrate_vs_SNR_arrayconfig}\vspace{-3mm}
  \end{figure}

We investigate the impact of different subarray distributions on the performance of \ac{isac} systems.
To evaluate the impact of array geometry, we compare three modular array configurations under fixed $K$ and $M$: random layout where subarrays are randomly distributed, uniform layout where subarrays are uniformly distributed, and collocated layout where subarrays are closely spaced with half-wavelength separation.
Fig.~\ref{Comrate_vs_SNR_arrayconfig} shows the communication spectral efficiency versus received SNR for proposed algorithms under different array configurations. 
It is noted that for both algorithms, the modular arrays with random and uniform layouts achieve similar performance and significantly outperform the collocated layout, with the performance gap widening at higher SNRs. 
This superiority of widely-spaced distributions stems from their enhanced near-field characteristics, leading to higher channel rank and improved spatial multiplexing gain.
Additionally, SDR-RRS consistently outperforms RM-JGD across all configurations, demonstrating its better capability in handling the non-convex optimization problem.

\begin{figure} [t]
\setlength{\abovecaptionskip}{-0.1cm}
		\centering	\includegraphics[width=0.35 \textwidth]{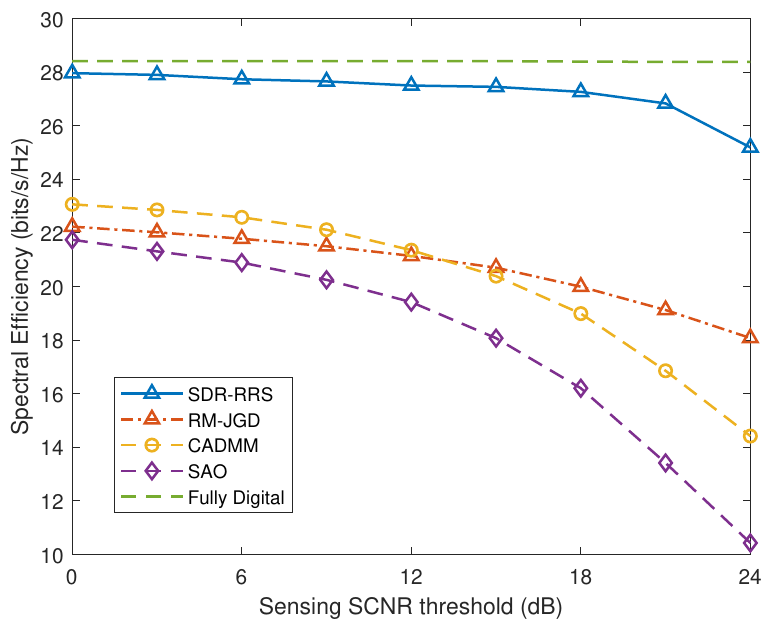}
		\caption{Communication spectral efficiency versus sensing SCNR threshold for different beamforming schemes.}
\label{benchmark}\vspace{-2mm}
  \end{figure}
We compare the proposed SDR-RRS and RM-JGD algorithms with three baselines: fully-digital beamforming (FDB) as the performance upper bound, subarray-based alternating algorithm (SAO) \cite{shen2022alter}, and consensus alternating direction method of multipliers (CADMM) \cite{chen2024mimodfrc}.
Fig.~\ref{benchmark} shows the communication spectral efficiency versus sensing \ac{scnr} threshold for different schemes with 28 \ac{rf} chains.
The results reveal that SDR-RRS consistently achieves the highest spectral efficiency among all hybrid beamforming schemes, closely approaching the FDB upper bound.
Additionally, CADMM outperforms RM-JGD at low SCNR thresholds, while SAO shows the lowest performance.
The SDR-RRS algorithm's superior performance stems from its ability to find a near-optimal solution by leveraging the optimal analog beamformer and approximating the global optimum, while the RM-JGD algorithm may converge to a local optimum far from the global optimum, resulting in slightly inferior performance.
In contrast, the SAO and CADMM algorithms alternately optimize the digital and analog beamformers and employ various approximations to the objective function and rely on various approximations,  which may lead to suboptimal solutions and further performance degradation.

\begin{figure} [t]
\setlength{\abovecaptionskip}{-0.1cm}

		\centering	\includegraphics[width=0.35 \textwidth]{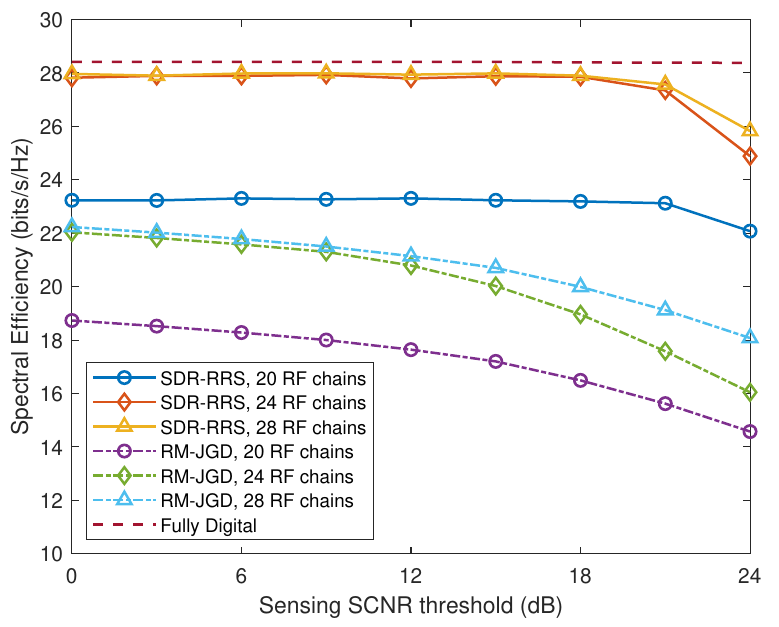}
		\caption{Communication spectral efficiency versus sensing SCNR threshold for different numbers of \ac{rf} chains.}
\label{RFchain_tradeoff}\vspace{-2mm}
  \end{figure}
We then present the communication spectral efficiency with different sensing SCNR thresholds in Fig.~\ref{RFchain_tradeoff} to show the trade-off between sensing and communication performance.
The performance of  fully digital beamforming  serves as an upper bound for comparison.
As the sensing SCNR threshold increases, the communication spectral efficiency decreases for both algorithms, with a more severe performance degradation observed for the case with fewer \ac{rf} chains. 
Specifically, for low sensing SCNR thresholds, the communication spectral efficiency of SDR-RRS with a hybrid beamformer using 24 \ac{rf} chains closely approaches the fully digital performance.
However, the performance gap widens at higher SCNR thresholds. 
Besides,  the SDR-RRS algorithm consistently outperforms the RM-JGD algorithm in terms of communication spectral efficiency across all SCNR thresholds and RF chain configurations, indicating its superior ability to strike a balance between communication and sensing performances.

\begin{figure} [t]
\setlength{\abovecaptionskip}{-0.1cm}
		\centering	\includegraphics[width=0.35\textwidth]{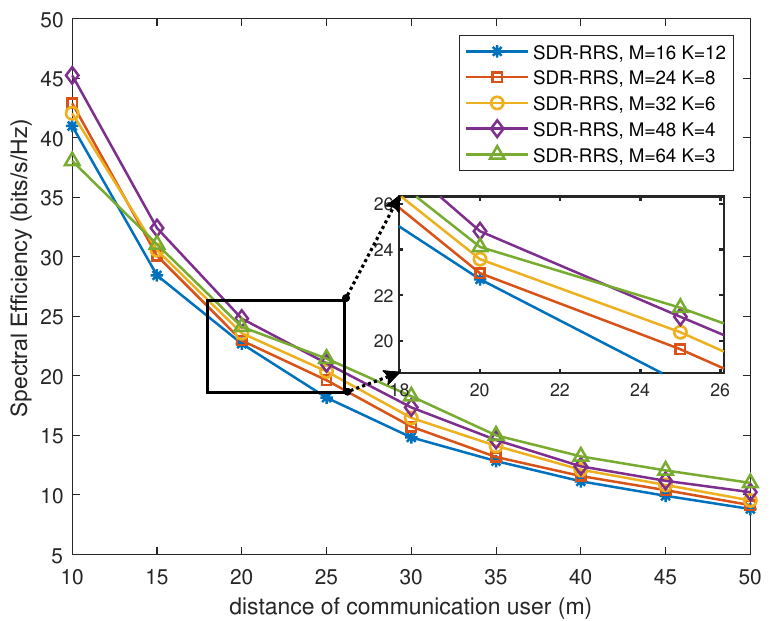}
		\caption{Communication spectral efficiency versus user distance for different subarray scales based on SDR-RRS algorithm.}
\label{subscale_vs_du}\vspace{-2mm}
  \end{figure}
  
  \begin{figure*}[t]
 \centering 
  \subfloat[M=32, K=8] {\label{fig:subfig1}\includegraphics[width=0.27\textwidth]{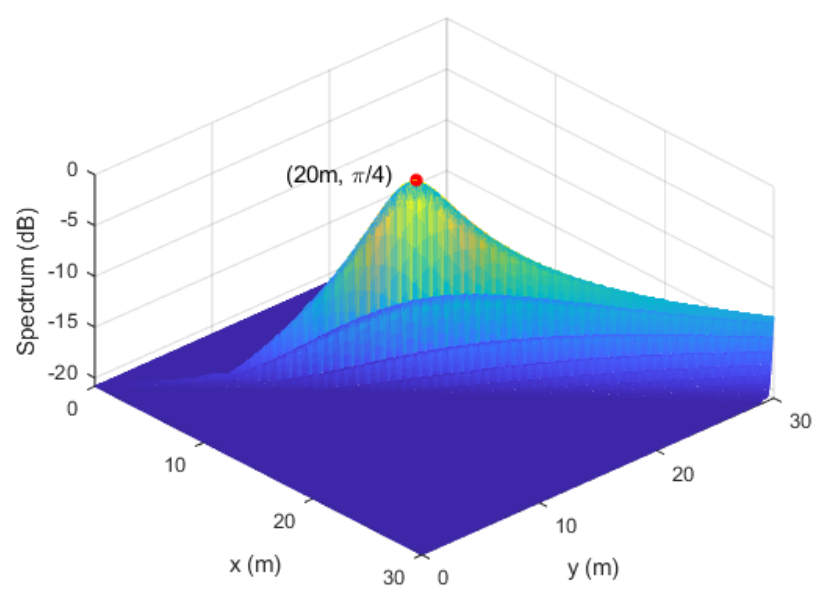}}
\subfloat[M=32, K=12] {\label{fig:subfig2}\includegraphics[width=0.27\textwidth]{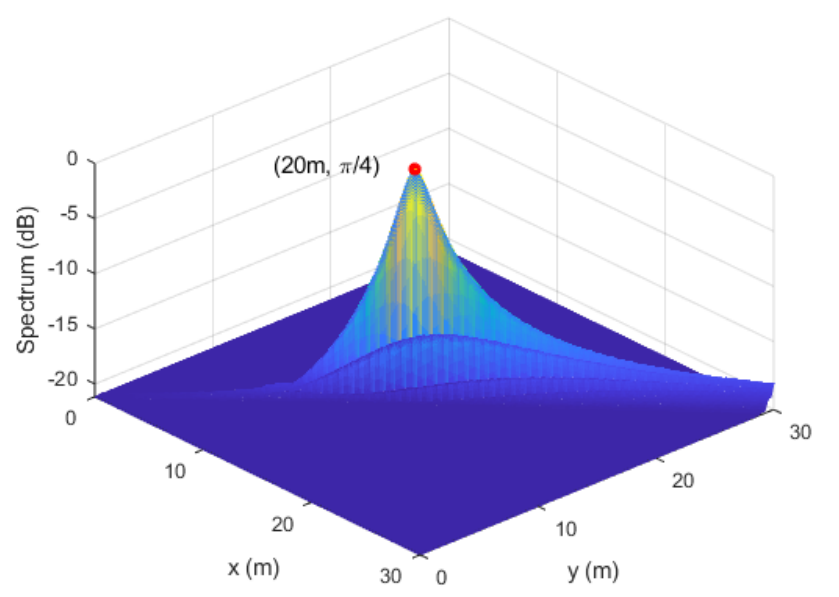}}
\subfloat[M=32, K=16] {\label{fig:subfig3}\includegraphics[width=0.27\textwidth]{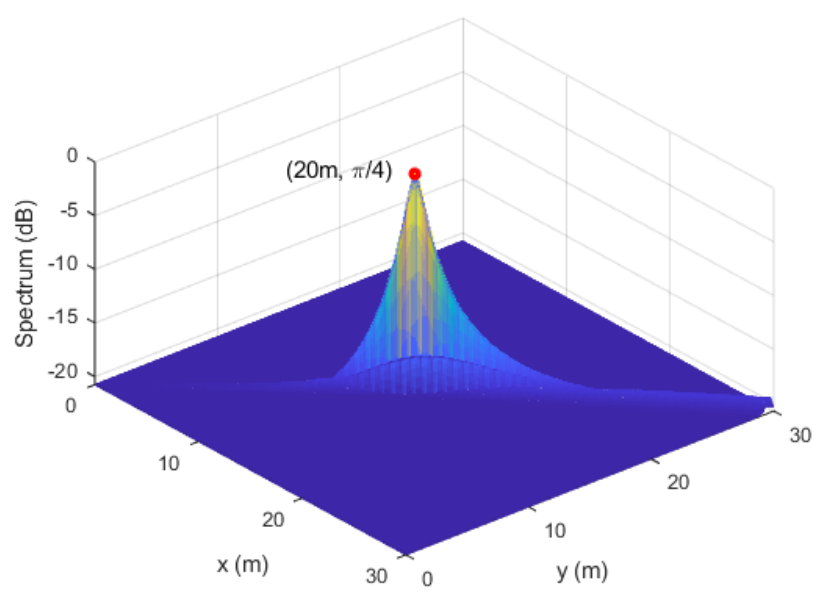}}
  \caption{Normalized \ac{music} spectrum for different numbers of subarrays $K$ with $M=32$ antennas per subarray.} 
   \label{MUSIC_specturm}\vspace{-2mm}
\end{figure*}  
Fig.~\ref{subscale_vs_du} investigates the impact of subarray scale on the communication performance of SDR-RRS algorithm.
 Under fixed total array aperture $S$ and antenna number $N_t = 192$, we vary the subarray antenna number $M$ as 16, 24, 32, 48, and 64, with corresponding subarray number $K$ being 12, 8, 6, 4, and 3.
It can be observed that the optimal subarray scale varies with the user distance. 
For users within 25 m, the subarray scale of $M = 48$ outperforms the others, while beyond 25 m, the subarray scale of $M = 64$ achieves the highest communication spectral efficiency.
A larger number of subarrays $K$ provides higher spatial multiplexing gains, but a smaller $M$ reduces the beamforming gain. %
Thus, the optimal subarray scale at different distances represents a trade-off between spatial multiplexing and beamforming gains for maximizing the communication performance. 
Moreover, the performance gap between subarray scales  decreases with distance, due to the transition from near-field spherical wavefronts to far-field planar wavefronts, which gradually diminishes the spatial multiplexing advantages of having more subarrays.

In Fig.~\ref{MUSIC_specturm}, the normalized \ac{music} spectrum obtained by the proposed SDR-RRS algorithm is compared for $M = 32$ with varying numbers of subarrays $K$, over a grid spanning $x \in [0:0.06:30]$ m and $y \in [0:0.06:30]$ m. 
The peak of the spectrum is consistently observed at the actual target location $(20 \mathrm{m}, \pi/4)$ for all values of $K$.  
Notably, the main lobe width containing the peak value narrows as $K$ increases, indicating enhanced range resolution at the same angular direction. 
This behavior can be attributed to the spherical wavefront characteristic across subarrays, which becomes more pronounced with increasing $K$, leading to improved range estimation accuracy. 
The results demonstrate that augmenting the number of subarrays can significantly enhance the range resolution capability of the \ac{music} algorithm, thereby enabling more precise target localization.

\section{Conclusion}\label{section_conclusion}
In this paper, we have developed the low-complexity hybrid beamforming algorithms for  modular \ac{xl-mimo} \ac{isac} systems. 
By exploiting the structural similarity between the communication and sensing channels based on the  the piecewise-far-field channel model, we have derived a closed-form solution for the optimal analog beamformer and transformed the joint analog-digital design problem into a lower-dimensional digital beamformer optimization.
Two approaches have been developed for the rank-constrained digital beamformer optimization: a manifold-based method and an \ac{sdr}-based method for obtaining near-optimal solutions.
Extensive simulations have validated the effectiveness of the proposed  algorithms.


	\appendices
    \section{Proof of Lemma \ref{Hc_rank}}\label{channel_rank_proof}
    First, since $\mathbf{H}_c \in \mathbb{C}^{N_c \times N}$ and $N_c \ll N$, the rank of $\mathbf{H}_c$ is upper bounded by $N_c$. 
  For a given propagation path $p$, the subarray response vectors $\mathbf{a}_{tp}^{1}\left( \theta _{tp}^{1} \right),\ldots,\mathbf{a}_{tp}^{K}\left( \theta _{tp}^{K} \right)$ are linearly independent.
		Similarly, due to the distinguishability of $N_p$ propagation paths, the angles of different paths are different.
  Therefore, for a fixed transmit subarray $k$, the response vectors $\mathbf{a}_{t1}^{k}\left( \theta _{t1}^{k} \right),\ldots,\mathbf{a}_{tN_p}^{k}( \theta _{tN_p}^{k})$ for different propagation paths $p=1,\ldots,N_p$ are linearly independent.
		Therefore, we can obtain the $n$-th row of the communication channel $\mathbf{H}_c$ in  (\ref{hspm_com}) $\mathbf{H}_c(n,:)$ is given by
  \begin{equation}
      \begin{aligned}
  & {{\mathbf{H}}_{c}}\left( n,: \right)=\\
  &\left[ \sum\limits_{p=1}^{{{N}_{p}}}{\mu _{p}^{1}} \mathbf{a}_{cp}^{1}\left( n \right){{\left( \mathbf{a}_{tp}^{1} \right)}^{H}} \right.\left. \cdots\  \sum\limits_{p=1}^{{{N}_{p}}}{\mu _{p}^{K}} \mathbf{a}_{cp}^{K}\left( n \right){{\left( \mathbf{a}_{tp}^{K} \right)}^{H}} \right],   
\end{aligned}
  \end{equation}
where $\mathbf{a}_{tp}^{k}$ is a shorthand notation for $\mathbf{a}_{tp}^{k}\left( \theta _{tp}^{k} \right)$, and $\mathbf{a}_{cp}^{k}\left( n \right)$ denotes the $n$-th element of $\mathbf{a}_{cp}^{k}\left( \theta_{cp}^{k} \right)$.
%
Therefore, each row vector $\mathbf{H}_c(n,:)$ is a linear combination of the following $KN_p$ linearly independent row vectors:
\begin{equation} \label{independent_vector}
    \mathbf{\bar{a}}_{tp}^{k} \triangleq [\  \underbrace{\mathbf{0}, \ldots , \mathbf{0}}_{k-1},(\mathbf{a}_{tp}^{k})^{H}, \underbrace{\mathbf{0}, \ldots , \mathbf{0}}_{K-k} \ ] \in \mathbb{C}^{1\times KM}, 
\end{equation}
where $k=1, 2, \cdots, K$, $p=1, 2, \cdots, N_p$, and $\mathbf{0}$ is an all-zero vector of dimension $1\times M$.

However, the arrival angles corresponding to different propagation paths from different transmit subarrays may not be spatially distinguishable at the receive array, resulting in potential linear dependency among the vectors in \eqref{independent_vector}. 
This dependency reduces the rank of the channel matrix $\mathbf{H}_c$, thereby establishing its upper bound as $KN_p$, which validates the right-hand inequality.
The left-hand inequality holds in the extreme case where all transmit subarrays share the same arrival angle to the receive array for a given propagation path $p$, i.e., $\theta_{cp}^1= \theta_{cp}^2=\ldots = \theta_{cp}^K$\cite{chen2023hybrid}.
	\section{Proof of Theorem \ref{optimal_structure_RX}} \label{optimal_structure_RX_proof}
	By letting $\mathbf{R}_{X}^{*}=\boldsymbol{\Delta} \boldsymbol{\Delta}^{\mathrm{H}}$, we can decompose $\boldsymbol{\Delta}$ as
	\begin{equation}
        \setlength{\abovedisplayskip}{3pt}
	\setlength{\belowdisplayskip}{3pt}
\boldsymbol{\Delta}=P_{\tilde{U}}\boldsymbol{\Delta}+\mathbf{P}_{\tilde{U}}^{\perp} \boldsymbol{\Delta},
	\end{equation}
    where $\mathbf{P}_{\tilde{U}}=\mathbf{{\tilde{U}}}(\mathbf{{\tilde{U}}}^H \mathbf{{\tilde{U}}})^{-1} \mathbf{{\tilde{U}}}^H$ denotes the orthogonal projection onto the subspace spanned by the columns of $\mathbf{{\tilde{U}}}$ in (\ref{tilde_U_k}), and $\mathbf{P}_{\tilde{U}}^{\perp}=\mathbf{I}-\mathbf{P}_{\tilde{U}}$.
    Furthermore, we can decompose $\boldsymbol{\Delta}$ additively as
    \begin{equation} \label{RX_decomposion}
        \setlength{\abovedisplayskip}{3pt}
	\setlength{\belowdisplayskip}{1pt}
\mathbf{R}_{X}^*=\mathbf{P}_{\tilde{U}} \boldsymbol{\Delta}\boldsymbol{\Delta}^H \mathbf{P}_{\tilde{U}}+\tilde{\mathbf{R}}_{X},
    \end{equation}
    with
    \begin{equation}\label{tilde_RX}
        \setlength{\abovedisplayskip}{3pt}
	\setlength{\belowdisplayskip}{3pt}
{{\widetilde{\mathbf{R}}}_{X}}=\mathbf{P}_{{\tilde{U}}}^{\bot }\mathbf{\Delta }{{\mathbf{\Delta }}^{H}}\mathbf{P}_{{\tilde{U}}}^{\bot }+{{\mathbf{P}}_{{\tilde{U}}}}\mathbf{\Delta }{{\mathbf{\Delta }}^{H}}\mathbf{P}_{{\tilde{U}}}^{\bot }+\mathbf{P}_{{\tilde{U}}}^{\bot }\mathbf{\Delta }{{\mathbf{\Delta }}^{H}}{{\mathbf{P}}_{{\tilde{U}}}}.
    \end{equation}
    By utilizing the property of orthogonal projection matrices that ${{\mathbf{P}}_{{\tilde{U}}}}\mathbf{\tilde{U}}=\mathbf{\tilde{U}}$ and ${{\mathbf{\tilde{U}}}^{H}}{{\mathbf{P}}_{{\tilde{U}}}}={{\mathbf{\tilde{U}}}^{H}}$, we can obtain that
    \begin{equation}
        \setlength{\abovedisplayskip}{3pt}
	\setlength{\belowdisplayskip}{3pt}
    	{{\mathbf{\tilde{U}}}^{H}}{{\widetilde{\mathbf{R}}}_{X}}\mathbf{\tilde{U}}=\boldsymbol{0}.
    \end{equation}

    According to (\ref{tilde_U}), utilizing the property of permutation matrix that $\mathbf{P}^{-1}=\mathbf{P}^T$, we have
    \begin{equation}\label{U_tilde_U}
        \setlength{\abovedisplayskip}{3pt}
	\setlength{\belowdisplayskip}{3pt}
    		\mathbf{U} =\mathbf{\tilde{U}}\mathbf{P}^T.
    \end{equation}
    Therefore, by substituting (\ref{U_tilde_U}) into (\ref{gt_linearcom}) and (\ref{Vc_linearcom}), $\mathbf{V}_c$ and ${{\mathbf{g}}_{tq}}, \forall q$ can be expressed as linear transformations of $\mathbf{\tilde{U}}$, namely
    \begin{equation} 
        \setlength{\abovedisplayskip}{3pt}
	\setlength{\belowdisplayskip}{3pt}\label{gt_linearcom_tilde_U} 	{{\mathbf{g}}_{tq}}=\mathbf{\tilde{U}}\mathbf{P}^T\mathbf{\tilde{ {\boldsymbol{\nu }}}}_{tq} , \  \forall q,
    \end{equation}
    \begin{equation} 
        \setlength{\abovedisplayskip}{3pt}
	\setlength{\belowdisplayskip}{3pt}\label{Vc_linearcom_tilde_U}
	\setlength{\belowdisplayskip}{1pt}
    	{{\mathbf{V}}_{c}} =\mathbf{\tilde{U}}\mathbf{P}^T\mathbf{\tilde{T}}.
    \end{equation}
    Thus, it can be readily verified that
    \begin{subequations}\label{gVc_RX_gVc}
    	\begin{align}
    		\label{gtRgt}
    			{{\mathbf{g}}_{tq}}^{H}{{\mathbf{\tilde{R}}}_{X}}{{\mathbf{g}}_{tq}})&=\widetilde{\boldsymbol{\nu }}_{tq}^{H}\mathbf{P}{{\mathbf{\tilde{U}}}^{H}}{{\mathbf{\tilde{R}}}_{X}}\mathbf{\tilde{U}}{{\mathbf{P}}^{T}}{{\widetilde{\boldsymbol{\nu }}}_{tq}}=0, \forall q,\\
    	    \label{VcRVc}
    	    \mathbf{V}_{c}^{H}{{\mathbf{\tilde{R}}}_{X}}{{\mathbf{V}}_{c}}&={{\mathbf{\tilde{T}}}^{H}}\mathbf{P}{{\mathbf{\tilde{U}}}^{H}}{{\mathbf{\tilde{R}}}_{X}}\mathbf{\tilde{U}}{{\mathbf{P}}^{T}}\mathbf{\tilde{T}}={{\mathbf{0}}}.
    	\end{align}   
    \end{subequations}
	Based on the \ac{svd} of ${{\mathbf{H}}_{c}}$, the achievable communication spectral efficiency in (\ref{comm_rate}) can be rewritten accordingly as
	\begin{equation} \label{svd_Comm_rate}
        \setlength{\abovedisplayskip}{3pt}
	\setlength{\belowdisplayskip}{3pt}
		C=\log \det \left( \mathbf{I}+\frac{1}{\sigma     _{c}^{2}}{\mathbf{\Sigma }_{c}}\mathbf{V}_{c}^{H}{{\mathbf{R}}_{X}}{{\mathbf{V}}_{c}} \right).
	\end{equation}
   
   By substituting (\ref{RX_decomposion}) into (\ref{sensing_scnr}) and (\ref{svd_Comm_rate}) and observing (\ref{gVc_RX_gVc}), we can conclude that the sensing \ac{scnr} and the achievable communication spectral efficiency are both independent of ${{\mathbf{\tilde{R}}}_{X}}$. 
   Besides, we note that
   \begin{equation}\label{power_RX}
       \setlength{\abovedisplayskip}{3pt}
	\setlength{\belowdisplayskip}{3pt}
   	\text{tr}\left( {{{\mathbf{\tilde{R}}}}_{X}} \right)=\text{tr}\left( \mathbf{P}_{{\tilde{U}}}^{\bot }\mathbf{\Delta }{{\mathbf{\Delta }}^{H}}\mathbf{P}_{{\tilde{U}}}^{\bot } \right)=\left\| \mathbf{\Delta }^H\mathbf{P}_{{\tilde{U}}}^{\bot } \right\|_{F}^{2}\ge 0,
   \end{equation}
    which means that the ${{\mathbf{\tilde{R}}}_{X}}$ component does not contribute to either the communication or sensing performance, but only consumes the transmit power\cite{li2008range}.
    Thus, we must have $\text{tr}( {{{\mathbf{\tilde{R}}}}_{X}})=0$ to satisfy the transmit power constraint.
    The equality in (\ref{power_RX}) holds if  and only if $\mathbf{\Delta }^H\mathbf{P}_{{\tilde{U}}}^{\bot }=\mathbf{0}$,  and it can be observed from (\ref{tilde_RX}) that this implies ${{\mathbf{\tilde{R}}}_{X}}=\mathbf{0}$.
    
    Therefore, the optimal covariance matrix $\mathbf{R}_X^{*}$ can be written as
    \begin{equation}
      \setlength{\abovedisplayskip}{0.5pt}
	\setlength{\belowdisplayskip}{0.5pt}
    	\mathbf{R}_{X}^{*}={{\mathbf{P}}_{{\tilde{U}}}}\mathbf{\Delta }{{\mathbf{\Delta }}^{H}}{{\mathbf{P}}_{{\tilde{U}}}}\triangleq \mathbf{\tilde{U}\Lambda }{{\mathbf{\tilde{U}}}^{H}},
    \end{equation}
   where $	\mathbf{\Lambda }$ is a positive semi-definite matrix, and it can be given by
   \begin{equation}
     \setlength{\abovedisplayskip}{0.5pt}
	\setlength{\belowdisplayskip}{0.5pt}
   	\mathbf{\Lambda }={{({{\mathbf{\tilde{U}}}^{H}}\mathbf{\tilde{U}})}^{-1}}{{\mathbf{\tilde{U}}}^{H}}\mathbf{\Delta }{{\mathbf{\Delta }}^{H}}\mathbf{\tilde{U}}{{({{\mathbf{\tilde{U}}}^{H}}\mathbf{\tilde{U}})}^{-1}},
   \end{equation}
   which completes the proof.

	\section{Proof of lemma \ref{structure_beamformer}} \label{optimal_structure_WBB_proof}
First, we denote  $\mathbf{A} \triangleq \mathbf{I}+\frac{1}{\sigma _{c}^{2}}\mathbf{W}_{\text{BB}}^H \mathbf{B}\mathbf{W}_{\text{BB}}$.
 To complete the main proof, we introduce Lemma~\ref{lemma_proof} \cite[Theorem 9.B.1]{marshall1979inequalities}.
 \begin{lemma} \label{lemma_proof} 
	\emph{ Let $\mathbf{R}$ be an $n\times n$ Hermitian matrix with diagonal elements denoted by the vector $\mathbf{d}$ and eigenvalues denoted by the vector $\boldsymbol{\lambda}$. Then, 
 \begin{equation}
     \mathbf{d} \prec \boldsymbol{\lambda},
 \end{equation}
 which means $\mathbf{d}$ is majorized by $\boldsymbol{\lambda}$\cite[Definition 1.A.1]{marshall1979inequalities}.
    }	
\end{lemma}

Since maximization of mutual information is Schur-concave, according to Lemma~\ref{lemma_proof}, we have 
 \begin{equation}
     \log \det (\boldsymbol{\lambda}(\mathbf{A}))\leq\log \det (\text{Diag}(\mathbf{A})),
 \end{equation}where $\boldsymbol{\lambda}(\mathbf{A})$ is the vector of eigenvalues of $\mathbf{A}$ in decreasing order.
Therefore, if the optimal $\mathbf{W}_{\text{BB}}^*$ does not diagonalize $\mathbf{B}$,  we can always find a unitary matrix $\mathbf{M}$ satisfying $\mathbf{M}\mathbf{M}^H = \mathbf{I}$,  which diagonalizes $\mathbf{B}$.
By defining $\tilde{\mathbf{W}}_{\text{BB}}^*=\mathbf{W}_{\text{BB}}^*\mathbf{M}$, the objective function \eqref{digital_objective} is diminished.
	
It is also noted that 
\begin{subequations}
   \begin{align}      &\operatorname{tr}\left({\tilde{\mathbf{W}}_{\text{BB}}^*} (\tilde{\mathbf{W}}_{\text{BB}}^*)^H\right) = \operatorname{tr}\left({\mathbf{W}_{\text{BB}}^*}  ({\mathbf{W}}_{\text{BB}}^*)^H\right) \leq \frac{N_s}{M},\\
    &\operatorname{tr}\left( (\tilde{\mathbf{W}}_{\text{BB}}^*)^H\boldsymbol{\Phi}_q\tilde{\mathbf{W}}_{\text{BB}}^*\right) =   \operatorname{tr}\left( ({\mathbf{W}}_{\text{BB}}^*)^H\boldsymbol{\Phi}_q{\mathbf{W}}_{\text{BB}}^*\right), \forall q,
   \end{align}
\end{subequations}
which implies that  constraints \eqref{cons_power_WBB} and \eqref{change_scnr_cons} still hold with $\tilde{\mathbf{W}}_{\text{BB}}^*$.
This means $\tilde{\mathbf{W}}_{\text{BB}}^*$ is a better solution than ${\mathbf{W}}_{\text{BB}}^*$, leading to a contradiction.  
Hence, there exists an optimal solution  ${\mathbf{W}}_{\text{BB}}^*$ which can diagonalize ${\mathbf{B}}$.

As a result,  $\mathbf{A}$ can also be diagonalized by ${\mathbf{W}}_{\text{BB}}^*$, yielding ${\mathbf{W}_{\text{BB}}^*}^H \mathbf{B}\mathbf{W}_{\text{BB}}^*=\mathbf{D}$, where $\mathbf{D}$ is a square diagonal matrix.
Consequently, we obtain
\begin{equation}\label{wbb_inter}
    \mathbf{B}^{\frac{1}{2}}\mathbf{W}_{\text{BB}}^* = 
 \mathbf{\tilde{V}}\tilde{\boldsymbol{\Sigma}},
\end{equation}
where $\mathbf{\tilde{V}}$ is an arbitrary unitary matrix, $\boldsymbol{\Sigma}=\left[\begin{array}{c}
\mathbf{D}^{1 / 2} \\
\mathbf{0}
\end{array}\right]$ and $ \boldsymbol{\Sigma}^H \boldsymbol{\Sigma}=\mathbf{D}$
Taking  \ac{svd} of $\mathbf{B}$, we have $\mathbf{B} = {{\mathbf{U}}_{\text{B}}}\mathbf{\Sigma }_{\text{B}}\mathbf{U}_{\text{B}}^{H}$.
Based on \eqref{wbb_inter}, it follows that
\begin{equation}
\begin{aligned}
    \mathbf{W}_{\text{BB}}^* = \mathbf{B}^{-\frac{1}{2}}\mathbf{\tilde{V}}\tilde{\boldsymbol{\Sigma}} ={{\mathbf{U}}_{\text{B}}}\mathbf{\Sigma }_{\text{B}}^{-\frac{1}{2}}\mathbf{U}_{\text{B}}^{H}\mathbf{\tilde{V}}\tilde{\boldsymbol{\Sigma}},
\end{aligned}
\end{equation}
which completes the proof.

	\bibliographystyle{IEEEtran}
	\bibliography{mybib}
\end{document}